\documentclass[format=acmsmall, review=false]{acmart}
\usepackage{acm-ec-22}

\newcommand{\agents}{\mathcal{N}} 
\newcommand{\ninput}{\mathcal{I}} 
\newcommand{\uu}{\bigcup_{U\in \mathbf{U^t}}U}
\newcommand{\upu}{\biguplus_{U\in \mathbf{U^t}}U}

\setcitestyle{acmnumeric}

\usepackage{multirow,array}

\usepackage{graphicx}
\graphicspath{{images/}}
\usepackage{amsmath}
\usepackage{amsthm}
\usepackage{csquotes}
\usepackage{color}
\usepackage{thm-restate}
\usepackage{thmtools}
\usepackage{caption}
\usepackage{subcaption}

\usepackage{algpseudocode}
\usepackage[linesnumbered,ruled,vlined,noend]{algorithm2e}

\newcommand{\agents}{\mathcal{N}} 
\newcommand{\ninput}{\mathcal{I}} 
\newcommand{\uu}{\bigcup_{U\in \mathbf{U^t}}U}
\newcommand{\upu}{\biguplus_{U\in \mathbf{U^t}}U}

\newtheorem{definition}{Definition}
\newtheorem{theorem}{Theorem}
\newtheorem{lemma}{Lemma}
\newtheorem{claim}{Claim}

\newtheorem{corollary}{Corollary}
\newtheorem{example}{Example}
\newtheorem{proposition}{Proposition}
\newtheorem{remark}{Remark}

\theoremstyle{remark}

\makeatletter
\newcounter{protocol}
\newenvironment{protocol}[1][H]
  {
   \let\c@algocf\c@protocol
   \begin{algorithm}[#1]%
  }{\end{algorithm}}
\makeatother
  
\makeatletter
\newcounter{strategytemplate}
\newenvironment{strategytemplate}[1][H]
  {
   \let\c@algocf\c@strategytemplate
   \begin{algorithm}[#1]%
  }{\end{algorithm}}
\makeatother



\title{Long-term Data Sharing under Exclusivity Attacks}

\author{Yotam Gafni 
\lowercase{yotam.gafni@campus.technion.ac.il} \\ Moshe Tennenholtz \lowercase{moshet@ie.technion.ac.il} \\
\\ 
Technion - Israel Institute of Technology}

\gdef\shortauthors{}

\begin{abstract}
    The quality of learning generally improves with the scale and diversity of data. Companies and institutions can therefore benefit from building models over shared data. Many cloud and blockchain platforms, as well as government initiatives, are interested in providing this type of service. 

These cooperative efforts face a challenge, which we call ``exclusivity attacks''. A firm can share distorted data, so that it learns the best model fit, but is also able to mislead others. We study protocols for long-term interactions and their vulnerability to these attacks, in particular for regression and clustering tasks. We conclude that the choice of protocol, as well as the number of Sybil identities an attacker may control, is material to vulnerability. 
\end{abstract}

\begin{document}

\begin{titlepage}

\maketitle

\end{titlepage}

\section{The Work in Context}
\subsection{Data Sharing among Firms}
In today's data-oriented economy \cite{data_economy}, countless applications are based on the ability to extract statistically significant models out of acquired user data. Still, firms are hesitant to share information with other firms \cite{data_sharing_economy,data_sharing_europe}, as data is viewed as a resource that must be protected. This is in tension with the paradigm of the \emph{Wisdom of the crowds} \cite{wisdom_crowds}, which emphasizes the added predictive value of aggregating multiple data sources. As early as 2001, the authors in \cite{banko2001scaling} (note also a similar approach in \cite{halevy2009unreasonable}) concluded that 
\begin{displayquote}``... a logical next step for the research community would be
to direct efforts towards increasing the size of annotated training collections, while
deemphasizing the focus on comparing different learning techniques trained only on small training corpora.''
\end{displayquote}

 Two popular frameworks to address issues arising in settings where data is shared are multi-party computation \cite{mpc} and differential privacy \cite{dwork2008differential}. However,
these paradigms are focused on addressing the issue of privacy (whether of the individual user or the firm's data bank), but do not answer the basic conundrum of sharing data with \textit{competing} firms: On one hand, cooperation enables the firm to enrich its own models, but at the same time enable other firms to do so as well. A firm is thus tempted to game the mechanism to allow itself better inference than other firms. We call this behavior \emph{exclusivity attacks}. Even if supplying intentionally false information could be a legal risk, the nature of data processing (rich with outliers, spam accounts, natural biases), allows firms to have ``reasonable justification'' to alter the data they share with others. 

In this work, we present a model of collaborative information sharing between firms. The goal of every firm is first to have the best available model given the aggregate data. As a secondary goal, every firm wishes the others to have a downgraded version of the model. An appropriate framework to address this objective is the \emph{Non-cooperative computation} (NCC) framework, introduced in \cite{ncc}. The framework was considered with respect to one-shot data aggregation tasks in \cite{ncc4ml}. 

\subsection{Open and Long-term Environments
}
In our work, we present a general communication protocol for collaborative data sharing among firms, that can be associated with any specific machine learning or data aggregation algorithm. The protocol possesses an \emph{online} nature, when any participating firm may send (additional) data points at any time. This is in contrast with previous NCC literature, which focuses on \emph{one-shot} data-sharing procedures. The long-term setting yields two, somewhat contradicting, attributes:
\begin{itemize}
    \item A firm may send multiple subsequent inputs to the protocol, using it to learn how the model's parameters change after each contribution. For an attacker, this allows better inference of the true model's parameters, without revealing its true data points, as we demonstrate in Example~\ref{ex:attack_ncc} below. 
    
    \item A firm is not only interested in attaining the current correct parameters of the model, but also has a future interest to be able to attain correct answers, given that more data is later added by itself and its competitors. This has a chilling effect on attacks, as even a successful attack in the one-shot case could result in data corruption. For example, a possible short-term attack could be for a firm to send its true data, attain the correct parameters, and then send additional garbage data. Since we do not have built-in protection against such actions in the mechanism (for reasons further explained in Remark~\ref{remark:false_name}), this would result in data corruption for the other firms. Nevertheless, if the firm itself is interested in attaining meaningful information from the mechanism in the future, it would be disincentivized to do so.
\end{itemize}

We now give an example demonstrating the first point. In \cite{ncc4ml}, the authors consider the problem of collaboratively calculating the average of data points. They show in their Theorem 4.6 and Theorem 4.7 that whether the number of different data points is known is essential to the truthfulness of the mechanism. When the number of data points is unknown, the denominator of the average term is unknown, and it is impossible for an attacker to know with certainty how to attain the true average from the average the mechanism reports given a false input of the attacker. 
We now show that in a model where it is possible to send multiple requests (in fact, two), it is possible to report false information and attain the correct average: 

\begin{example}
\label{ex:attack_ncc}
Consider a firm with some data points $D_I$ with a total sum $S_I$ and number of points $N_I$. Other firms have data points $D_O$ with a total sum $S_O$ and number of points $N_O$. Assume $S_I \neq 0, N_I = 2$.\footnote{These assumptions are not \textit{required} for the attack scheme to succeed, but make for a simpler demonstration.} Instead of reporting $D_I$, the firm first reports $D' = [0]$, receives an average $a_1$, then reports $D'' = [0]$ and receives the updated average $a_2$. The average that others, following the mechanism as given, attain is $\frac{S_O}{N_O + 2}$, the true average is $\frac{S_I + S_O}{N_I + N_O}$, and they are different by our assumption on $S_I, N_I$. 
The firm is thus successful in misleading others. Moreover, the firm can infer the true average. Given $$a_1 = \frac{S_O}{N_O + 1} \neq  \frac{S_O}{N_O + 2} = a_2,$$
the firm\footnote{The only case where $a_1 = a_2$ is when $S_O = 0$. In this case, upon having $a_1 = 0$, we can choose $D'' = [1]$, and a similar argument shows that we can infer the true average.} can calculate
$$ N_O = \frac{a_1 - 2a_2}{a_2 - a_1}, S_O = a_1 (N_O + 1),$$
and thus have all the information required to calculate the true average. 
\end{example}

\begin{remark}
\label{remark:false_name}
{\it Why should we not consider simply forbidding multiple subsequent updates by a firm?}  As noted in \cite{yokoo,vcg_sybil,duplication_cheating}, modern internet-based environments lack clear identities and allow for multiple inputs by the same agent using multiple identities. A common distinction in blockchain networks separates \emph{public} (``permissionless'') and \emph{private} (``permissioned'') networks \cite{permissioned_vs_permissionless}, where \emph{public} networks allow open access for everyone, while \emph{private} networks require additional identification for participation. In both cases, however, it is impossible to totally prevent false-name manipulation, where a firm uses multiple identities to send her requests. Therefore, any ``simple'' solution of the problem demonstrated in Example~\ref{ex:attack_ncc} is impossible. The mechanism does not know whether multiple subsequent updates are really sent by different firms, or they are in fact ``sock puppets'' of a single firm. The mechanism therefore can not adjust appropriately (e.g., drop any request after the first one). In this work, we assume a firm may control up to $\ell$ identities, and so in the formal model, we allow up to $\ell$ subsequent updates of a single firm. The false identities are not part of the formal model: They instead are encapsulated by giving firms this ability to update $\ell$ times subsequently. 
\end{remark}

\subsection{Our Results}

\begin{itemize}

    \item 
    
We define two long-term data-sharing protocols (the continuous and periodic communication protocols) for data sharing among firms. The models differ in how communication is structured temporally (whether the agents can communicate at any time, or are asked for their inputs at given times). Each model can be coupled with any choice of algorithm to aggregate the data shared by the agents. 
    
    \item We give a condition for NCC-vulnerability of an algorithm (given the communication model) in Definition~\ref{def:l_ncc}. A successful NCC attack is one that (i) Can mislead the other agents, and (ii) Maintains the attacker's ability to infer the true algorithm output. We give a stronger condition of NCC-vulnerability* that can moreover (i*) Mislead the other agents in every possible scenario. 
    As a simple example of using these definitions, we show in Appendix~\ref{app:max_illustration} that finding the maximum over agent reports is NCC-vulnerable but not NCC-vulnerable*.  
    
        \item For the $k$-center problem, we show that it is vulnerable under continuous communication but not vulnerable under periodic communication. Moreover, we show that it is not vulnerable* even in continuous communication, using a notion of \emph{explicitly-lying} attacks. 
    
    \item For Multiple Linear Regression, we show that it is vulnerable* under continuous communication but not vulnerable under periodic communication. The vulnerability* in continuous communication depends on the number of identities an attacker can control: We show a form of attack so that an attacker with $d + 2$ identities (where $d$ is the dimension of the feature space) is guaranteed to have an attack, and an attacker with less than $d - 2$ identities can not attack. 
    
\end{itemize}

The vulnerability(*) results for the continuous communication protocol are summarized in Table~\ref{tab:vulnerability_results}. Both algorithms are not vulnerable(*) under the periodic communication protocol.

  \begin{table}
    \setlength{\extrarowheight}{2pt}
    \begin{tabular}{cc|c|c|}
      & \multicolumn{1}{c}{} & \multicolumn{1}{c}{Vulnerable}  & \multicolumn{1}{c}{Vulnerable*} \\\cline{3-4}
      \multirow{2}*{}  & $d$-LinearRegression & Yes, for any $\ell \geq 1$ & \shortstack{$\begin{cases} \text{Yes} & \ell \geq d+2 \\
      \text{No} & \ell \leq d - 2
      \end{cases}$}\\\cline{3-4}
      & $k$-Center & Yes, for any $\ell \geq 1$ & No \\\cline{3-4}
    \end{tabular}
    \caption{A summary of vulnerability(*) results in the continuous communication protocol.}
    \label{tab:vulnerability_results}
  \end{table}

We overview related work in Appendix~\ref{app:related_work}. 

\section{Model and Vulnerability Notions}

We consider a system where agents receive factual updates containing data points or states of the world. The agents apply their reporting strategy, performing ledger updates. Upon any ledger update, the ledger distributes the latest aggregate parameter calculation using $\rho$, the computation algorithm. 

Formally, let $[n] = \{1,\ldots, n\}$ be a set of $n$ agents. An update $U$ is of some type, depending on the computational problem. An update with metadata $\hat{U} = <j, t, U>$ complements an update $U$ with an agent $j\in \agents$, and a type $t \in \{Factual, Ledger\}$, where ``Factual'' updates represent a factual state of nature observed by an agent, and ``Ledger'' updates are what the agent shares with the ledger, which may differ from what she factually observes. We note that the ledger (which for simplicity we assume is a centralized third party) does not make the data public, but only shares the algorithm's updated outputs according to the protocol's rules. The computation algorithm $\rho(\mathbf{U^t})$ is an algorithm that receives a series of updates $\mathbf{U^t} = (U_1,...,U_t)$ of any length $t$ and outputs a result. In the continuous communication protocol, we have that algorithm outputs are shared with all agents upon every ledger update. 

In this section and Sections~\ref{sec:kcenter}-\ref{sec:lr} we focus on the continuous communication protocol.     
The continuous communication protocol simulates a system where agents may push updates at any time, initiated by them and not by the system manager. We model this by allowing them to respond to any change in the state of the system, including responding to their own ledger updates. The only limit to an agent endlessly sending updates to the ledger is that we restrict it to update at most $\ell$ times subsequently. The continuous communication protocol is a messaging protocol between nature, the agents, and the ledger. A particular protocol run is instantiated with nature-input $\ninput$, which is a series of some length $|\ninput|$ with each element being of the form $<j,U>$, which is a tuple comprised of agent $j\in \agents$ and an update $U$.

\begin{figure}
    \hspace{-1cm}\includegraphics[scale=0.4]{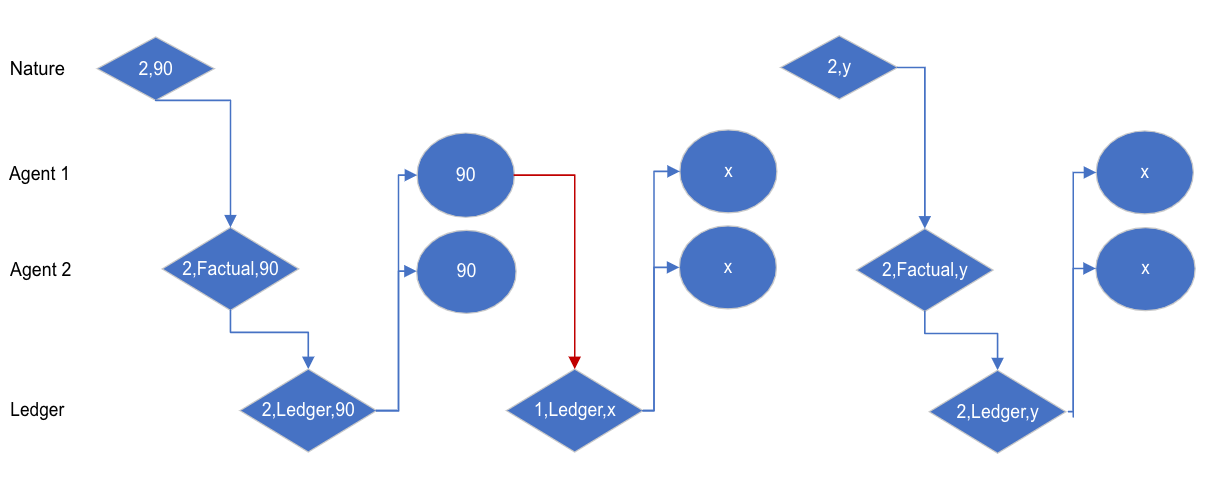}
    \caption{A continuous protocol run for $\ninput = (<2,90>, <2,y>)$ with some $90 < y < x$ and the algorithm $\rho = \max$, as explained in the proof of Proposition~\ref{prop:max_vulnerable*}. An agent's observed history are all the nodes in her line, or nodes that have an outgoing edge from a node in her line.}
    \label{fig:example_continuous}
\end{figure}

\begin{protocol}
\SetNoFillComment
\SetAlgoLined
\DontPrintSemicolon
\KwIn{Nature-input $\ninput$, Parameter $\ell$ the maximum number of subsequent updates by an agent}
\KwOut{Full Messaging History}

\For{ factual message $<j,U_{fact}>$ in $\ninput$} {

    Nature sends a message to $j$ with $<j,Factual,U_{fact}>$;
    
    activeMessage $\leftarrow$ True; \tcp{There is an active message}

    \While {activeMessage = True } {
        \tcc{As long as some agent is responding}

        activeMessage $\leftarrow$ False;
        
        \For{ agent $i:= 1$ to $n$} {
            \If{agent $i$ wishes to send a ledger update $U_{ledg}$ and last $\ell$ updates are not all of type $<i, Ledger, U>$\footnotemark } {
            
                $i$ sends a message to Ledger with $<i, Ledger, U_{ledg}>$;
                
                Ledger sends a message to all with $\rho$'s algorithm output over all the past ledger updates;      
                
                activeMessage $\leftarrow$ True;
                }
        }
    } 
}

 \caption{The continuous communication protocol}
 \label{proto:continuous}
\end{protocol}
\footnotetext{We can perhaps question whether agent $i$ respects the condition that not all of the last $\ell$ updates are not all of type $<i, Ledger, U>$ for some $U$. If she does not, she may send a message regardless of this constraint. But since nature can choose not to accept/respond to it, we simplify the protocol by assuming the agents self-enforce the constraint.}

For the analysis, we extract some useful variables from the run of the protocol that will be used in subsequent examples and proofs. 

Let a run $R$ be all the messages sent in the system during the application of the continuous communication protocol with nature-input $\ninput$ (where messages sent to 'all' appear once, and the messages appear in their order of sending).

Let $L_j(R), F_j(R)$ be the sub-sequences of all ledger, factual updates respectively in $R$ of agent $j$ (if the index $j$ is omitted, then simply all such updates, regardless of an agent). Let $O_j(R)$ (``observed history'' of $j$) be all the messages in $R$ received or sent by $j$ during the run of the nature protocol: These are all factual updates of $j$, ledger updates by $j$, and algorithm outputs shared by the ledger. Let $O_j(R)_{i_1:i_2}$ be the the elements of $O_j(R)$ starting with index $i_1$ and until (and including) index $i_2$. 

An update strategy for $j$ is a mapping $s_j$ from an observed history $O_j(R)$ to a ledger update $U_{ledg}$ by agent $j$.   
The truthful update strategy $truth_j$ is the following:
If the last element in $O_j(R)$ is of type $<j, Factual, U>$, update with $<j, Ledger, U>$. Otherwise, do not update. 

A full run of the protocol with nature input $\mathcal{I}$ and strategies $s_1,...,s_n$ is the run after completion of the nature protocol where nature uses input $\mathcal{I}$ and each agent $j$ responds using strategy $s_j$. Since we're interested in the effect of one agent deviating from truthfulness, we say that we run nature-input $\mathcal{I}$ with strategy $s_j$, where $j$ is the deviating agent, and it is assumed that all other agents $i\neq j$ play $truth_i$. We denote the resulting run $R_{\ninput, s_j}$. 

We can now define an NCC-attack on the nature protocol given algorithm $\rho$ and updates restriction $\ell$. 

\begin{definition}
\label{def:l_ncc}
An algorithm $\rho$ is $\ell-NCC-vulnerable$ if there exists an agent $j$ and update strategy $s_j$ such that:

i) There is a full run $R_{\ninput, s_j}$ of the protocol with some nature-input $\mathcal{I}$ and the strategy $s_j$ such that its last algorithm output is different from the last algorithm output in $R_{\ninput, truth_j}$. 

ii) For any two nature-inputs $\ninput, \ninput'$ such that 
the observed histories satisfy 
$$O_j(R_{\ninput,truth_j}) \neq  O_j(R_{\ninput',truth_j}) \implies O_j(R_{\ninput,s_j}) \neq O_j(R_{\ninput',s_j}).$$ 
\end{definition}

In words, to consider strategy $s_j$ as a successful attack, the first condition requires that there is a case where the rest of the agents other than $j$ observe something different than the factual truth. Notice that we strictly require that the \textit{other agents} (and not only the ledger) observe a different outcome: If $s_j$ updates with a ledger update that does not match its factual update, but this does not affect future algorithm outputs, we do not consider it an attack (It is a ``Tree that falls in a forest unheard''). The second condition requires that the attacker is always able to infer (at least in theory) the last true algorithm output. Under NCC utilities (which we omit formally defining, 
and work instead directly with the logical formulation, similar to Definition~1 in \cite{ncc}), failure to infer the true algorithm output under strategy $s_j$ makes it worse than $truth_j$, no matter how much the agent manages to mislead others (which is only its secondary goal). 

We remark without formal discussion that being $\ell$-NCC-vulnerable is enough to show that truthfulness is not an ex-post Nash equilibrium if the agents were to play a non-cooperative game using strategies $s_j$ with NCC utilities. However, it does not suffice to show that truthfulness is not a Bayesian-Nash equilibrium, as the cases where the deviation from truthfulness $s_j$ satisfies condition $(i)$ may be of measure 0. We give a stronger definition we call $\ell$-NCC-vulnerable*, that would guarantee the inexistence of the truthful Bayesian-Nash equilibrium for any possible probability measure, by amending condition $(i)$ to hold for \emph{all} cases:

\begin{definition}
An algorithm $\rho$ is $\ell-NCC-vulnerable*$ if there exists an agent $j$ and update strategy $s_j$ with both condition $(ii)$ of Definition~\ref{def:l_ncc}, and:

i*) For every full run $R_{\ninput, s_j}$ of the protocol with some nature-input $\mathcal{I}$, the last algorithm output is different than the last algorithm output in $R_{\ninput, truth_j}$. 

\end{definition}

As long as there is at least one full run of the protocol, it is clear that being $\ell$-NCC-vulnerable* implies being $\ell$-NCC-vulnerable. Similarly being $\ell$-NCC-vulnerable(*) implies being $(\ell+1)$-NCC-vulnerable(*) (i.e., the implication works for both the vulnerable and vulnerable* cases). 

In Appendix~\ref{app:max_illustration}, we illustrate the difference between the two definitions, as well as simple proof techniques, using a simple algorithm. 

\section{\texorpdfstring{$k$}---Center and \texorpdfstring{$k$}---Median in the Continuous Communication Protocol}
\label{sec:kcenter}

In this section, we analyze the performance of prominent clustering algorithms in terms of our vulnerability(*) definitions. Together with Section~\ref{sec:lr} this demonstrates the applicability of the approach for both unsupervised and supervised learning algorithms.

\begin{definition}
k-center: Each agent's update $U$ is a set of data points, where each data point is of the form $x \in \mathcal{R}^d$. A possible output of the algorithm is some $k$ centers that are among the data points $x_1,\ldots, x_k \in \uu$. Let $\eta_i = \{x | \arg \min_{j=1}^k ||x-x_j||_p = i\}_{x \in \uu}$ for $1 \leq i \leq k$ and some $L_p$ norm function $||\textbf{v}||_p = \sqrt[p]{\sum_{i=1}^d v_i^p}$ with $p\geq 1$. In words, $\eta_i$ is the set of all agents that have $x_i$ as their closest point among $x_1,\ldots, x_k$. Let $C(x_1,\ldots, x_k) = \max_{i=1}^k \max_{x\in \eta_i} ||x-x_i||$ be the cost function. In words, the cost of a possible algorithm output $x_1, \ldots, x_k$ is the maximum distance between a point and a center it is attributed to. We then have

\begin{equation}
\label{eq:kcenter_cost_function}
    \rho_{k-center}(\mathbf{U^t}) = \arg \min_{x_1,\ldots, x_k \in \uu} C(x_1,\ldots, x_k),
    \end{equation}
\end{definition}
i.e., the $k$ centers are the $k$ points among the reported points that minimize the cost if chosen as centers. Ties (both when determining $\eta_i$ and the final $k$ centers) are broken in favor of the candidate with the smallest norm\footnote{If this is not enough to determine, complement it with some arbitrary rule, e.g. over the radian coordinates of the points: This does not matter for the argument.}. 

\subsection{Sneak Attacks and Vulnerability}

In this subsection, we present a template for a class of attacks. We then show it is successful in showing the vulnerability of the protocol for $k$-center. 

\begin{figure}
    \includegraphics[scale=0.6]{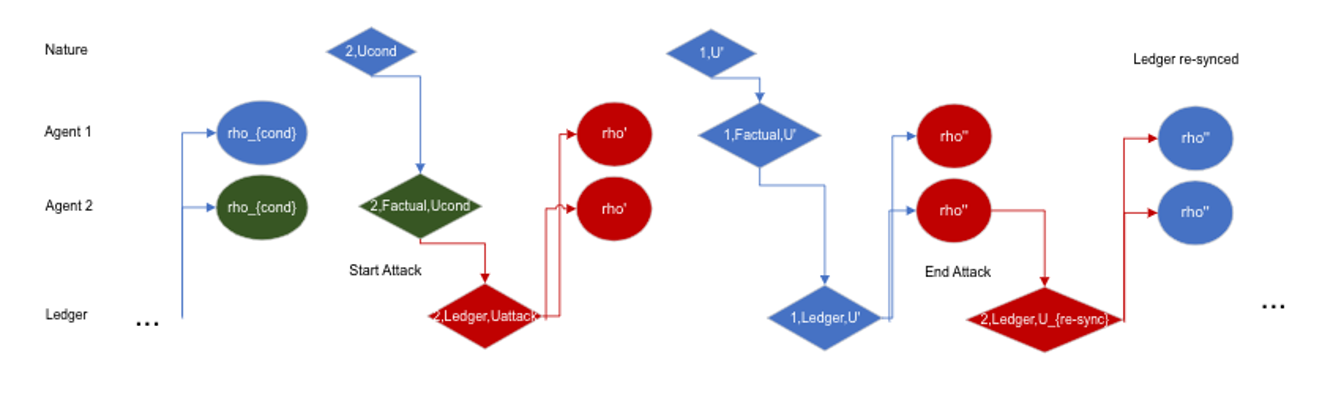}
    \caption{A general template for the sneak attack. Until the special conditions are met, and after the re-sync is done, the strategy behaves as $truth_j$.}
    \label{fig:sneak_attack_template}
\end{figure}

\begin{strategytemplate}
\SetAlgoLined
\SetNoFillComment
\DontPrintSemicolon
\KwIn{Observed history $O_j$. Parameters $U_{cond}, \rho_{cond}, U_{attack}, U_{re-sync}$}
\KwOut{A ledger update $<j,Ledger,U>$}

\tcc{Condition to start attack}

 \If{The last element in $O_j$ is $<j, Factual, U_{cond}>$, the last algorithm output in $O_j$ is $\rho_{cond}$, and the condition to start attack was not invoked before}{
    Return $<j, Ledger, U_{attack}>$
 }
 
 \tcc{Condition to end attack}
 \ElseIf{The condition to start attack was invoked, after that some agent (either $j$ or another) received a factual update, but the condition to end attack was not yet invoked} {
 
 Let $U$ be the last update in $O_j$ if it is a factual update for $j$, or $\emptyset$ otherwise. 
 
 Return $<j, Ledger, U \cup U_{re-sync}>$
 
 }
 
 \tcc{If the special conditions do not hold, act as $truth_j$}
 \ElseIf{Last update $U$ in $O_j$ is factual for $j$} {
    Return $<j, Ledger, U>$
 }
 
 \caption{A template for a \emph{sneak attack}}
 \label{alg:sneak_attack_template}
\end{strategytemplate}

Notice that when we defined strategies, we required them to be memory-less, i.e., only observe $O_j$ and not their own past behavior (which by itself anyway only depends on the past observed histories, which are contained in $O_j$). However, the conditions in Strategy Template~\ref{alg:sneak_attack_template} require for example to check whether the attack was initiated before. The technical lemma below shows that this is possible to infer from $O_j$. 

\begin{restatable}[]{lemma}{sneakAttackImp}
\label{sneakAttackImp}
If $U_{cond} \neq U_{attack}$, the sneak attack is well defined, i.e., the conditions to start and end attack can be implemented using only $O_j$.
\end{restatable}

We defer the proof details to Appendix~\ref{app:st_lemmas}. 

Strategy Template~\ref{alg:sneak_attack_template} presents the general \emph{sneak attack} form, which requires four parameters: $U_{cond}$, $\rho_{cond}$, the factual update and last algorithm output that serve as a signal for the attacker to send $U_{attack}$ - the deviation from truth performs, and $U_{re-sync}$, the update returning the ledger to a synced state. 

Two properties are important for a successful sneak attack. First, the attacker must know with certainty the algorithm output given the counter-factual that it would have sent $U_{cond}$ (as $truth_j$ would have), rather than $U_{attack}$. Second, after sending both $U_{attack}$ and $U_{re-sync}$, it should hold that all future algorithm outputs are the same as if sending only $U_{cond}$. For example, if updates are sets of data points and the algorithm outputs some calculation over their union (later formally defined in Definition~\ref{def:set_algs} as a \emph{set algorithm}), this holds if
$U_{cond} = U_{attack} \cup U_{re-sync}$.  

We formalize this intuition in the following lemma:

\begin{restatable}[]{lemma}{sneakAttackFormal}
\label{lem:sneak_attack_formal}
A sneak attack where $U_{attack} \subseteq U_{cond}, U_{re-sync} = U_{cond} \setminus U_{attack}$, and that moreover can infer the last algorithm output in $R_{\ninput, truth_j}$ after starting the attack and sending $U_{attack}$, satisfies condition $(ii)$. 
\end{restatable}

The proof of the lemma is given in Appendix~\ref{app:st_lemmas}.

We now give a sneak attack for $k$-center in $\mathcal{R}$. The example can be extended to a general dimension $\mathcal{R}^d$ by setting the remaining coordinates in the attack parameters to $0$.

\begin{example}
\label{ex:kcenter_omission_attack}
\emph{$k$-center with $k\geq 3$ is $1$-NCC-vulnerable using a sneak attack}: Use Strategy Template~\ref{alg:sneak_attack_template} with $U_{cond} = \{1,2,10,\ldots, 10^{k-1}\}, U_{attack} = \{1\}, U_{re-sync} = U_{cond} \setminus U_{attack}, \rho_{cond} = \{-\epsilon,0,\frac{\epsilon}{k-2}, \frac{\epsilon}{k-3}, \ldots, \epsilon\}$, with say $\epsilon = \frac{1}{1000}$. 

Condition $(i)$ is satisfied for nature-input $\ninput = (<1, , \rho_{cond}>, <2, U_{cond}>)$. The run with $truth_2$ yields algorithm outputs $\rho_{cond},\{1,10,\ldots, 10^{k-1}\}$ but the run with $s_2$ yields $\rho_{cond},\rho_{cond} \setminus \{\frac{\epsilon}{k-2}\} \cup \{1\}$. 

As for condition $(ii)$: Let $\ninput$ be some nature-input, and let $t$ be the index of the element of $\ninput$ after which the algorithm outputs $\rho_{cond}$ (i.e., $t+1$ is $<2,U_{cond}>$, upon where agent $2$ starts the attack). Let $M_+ = \max_{x\in \uu }$, $M_- = \min_{x\in \uu}$. Assume for simplicity that $|M_+| \geq |M_-|$, otherwise a symmetric argument to the one we lay out follows. Given the algorithm output  $\rho_{cond}$, we know that $\epsilon$ is the closest center to $M_+$. Thus, $M_+ - \epsilon \leq C(-\epsilon, 0, \epsilon) \leq C(M_-, 0, M_+) \leq \frac{M_+}{2}$. The last inequality is due to that every point is either in $[M_-,0]$ or $[0,M_+]$, and so its distance from the closest center is at most $\frac{1}{2}\max \{|M_-|, |M_+|\} = \frac{M_+}{2}$. We thus have that $M_+ \leq 2\epsilon$ (as illustrated in Figure~\ref{fig:kcenter_omission}). 

\begin{figure}
    \hspace{-2cm}
    \begin{subfigure}[b]{0.48\textwidth}
        \centering
        \includegraphics[scale=0.3]{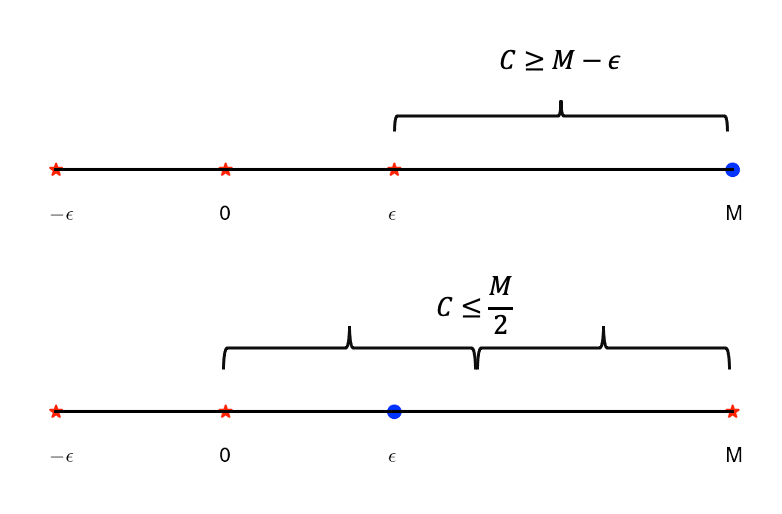}
        \caption{}
         \label{subfig:ko1}
     \end{subfigure}
    \begin{subfigure}[b]{0.48\textwidth}
    \includegraphics[scale=0.3]{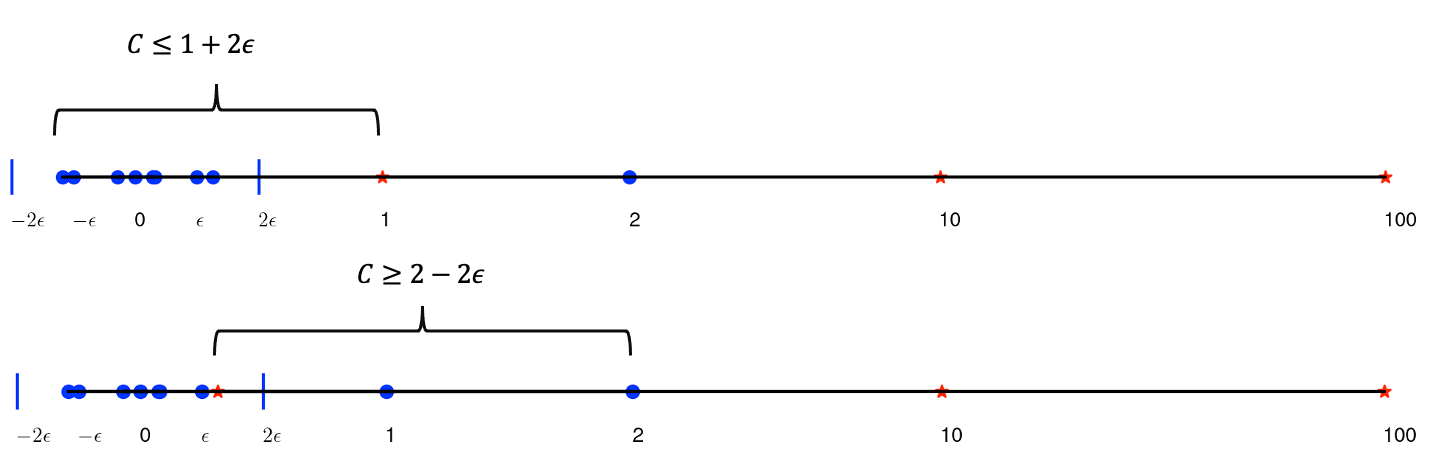}
        \caption{}
         \label{subfig:ko2}
     \end{subfigure}

    \caption{An illustration of Example~\ref{ex:kcenter_omission_attack} with $k=3$. In (a), the fact that $-\epsilon, 0, \epsilon$ is the algorithm output is enough to show that all input elements are within $[-2\epsilon, 2\epsilon]$, otherwise $M$ would be a better choice for a center. In (b), which is displayed on a logarithmic scale, we see that given that all prior input elements are within $[-2\epsilon, 2\epsilon]$, and with additional elements $1,2,10,100$, the algorithm must output $\{1,10,100\}$ as centers for a small enough $\epsilon$.}
    \label{fig:kcenter_omission}
\end{figure}

Therefore, under $truth_2$, after agent $2$ sends $U_{cond} = \{1,2,10,\ldots, 10^{k-1}\}$, we have $C(\{1,10,\ldots, 10^{k-1}\}) \leq 1 + 2\epsilon$. For any other choice of $k$ centers $x_1,\ldots,x_k$ (that may partially intersect), we have $C(\{x_1,\ldots ,x_k\}) \geq 2 - 2\epsilon$ (as illustrated in Figure~\ref{fig:kcenter_omission}). Choosing $\epsilon < \frac{1}{4}$ we have that the algorithm output must be $\{1,10,\ldots, 10^{k-1}\}$. This shows that agent $2$ can infer with certainty the algorithm output under $truth_2$. We thus satisfy the conditions of Lemma~\ref{lem:sneak_attack_formal}, which guarantees condition $(ii)$ is satisfied. 

\end{example}

\subsection{\texorpdfstring{$k$}---Center Vulnerability*}

In the previous subsection, we have shown that $k$-Center is vulnerable. However, in this subsection, we show it is not vulnerable*.

We note that a significant property of the $k$-center algorithm is that its output is a subset of its input. 

\begin{definition}
\label{def:set_algs}
A \emph{set} algorithm is an algorithm where each update is a set, and the algorithm is defined over the union of all updates $S = \uu$.

A \emph{multi-set} algorithm is an algorithm where each update is a multi-set of data points, and the algorithm is defined over the sum of all updates $S = \upu$.

A \emph{set-choice} algorithm is a set algorithm that satisfies
$\rho(S) \subseteq S$, i.e., the algorithm output is a subset of the input. 
\end{definition}

Many common algorithms such as max, min, or median, are set-choice algorithms, as well as $k$-center and $k$-median that we discuss. 

We notice a property of the sneak attack in Example~\ref{ex:kcenter_omission_attack}: $U_{attack}$ deducts points that exist in the factual update $U_{cond}$ and does not include them in the ledger update. In fact, throughout the run of $s_2$ the union of ledger updates by agent $2$ is a subset of the union of its factual updates. 
This leads us to develop the following distinction. 
We partition the space of attack strategies (all attacks, not necessarily just sneak attacks) into two types, explicitly-lying attacks and omission attacks. This distinction has importance beyond the technical discussion, because of legal and regulatory issues. Strategic firms may be willing to omit data (which can be excused as operational issues, data cleaning, etc), but not to fabricate data. 

Formally, for set and multi-set algorithms, we can partition all non-truthful strategies in the following way:

\begin{definition}

An \emph{explicitly lying} strategy $s_j$ is a strategy that for some nature-input $\mathcal{I}$ has a point $x \in L_j(R_{\ninput, s_j}), x \not \in F_j(R_{\ninput, s_j})$, i.e., the strategy  sends a ledger update with a point that does not exist in the union of all factual updates for that agent.

An \emph{omission} strategy $s_j$ is a 
a strategy that satisfies condition $(i)$ (i.e., misleads others)
that is not explicitly-lying.

For an omission strategy it must hold that for every run the agent past ledger updates are a subset of its factual updates, i.e., $L_j(R_{\ninput, s_j}) \subseteq F_j(R_{\ninput, s_j})$. 
\end{definition}

We now use the notion of explicitly-lying strategy to prove that $k$-center and $k$-median are not vulnerable*. For this we need one more technical notion:

\begin{definition}
\label{def:forceable_winners}
A set-choice algorithm has \emph{forceable winners} if for any set $S$ and a point $x\in S$, there is a set $\bar{S}$ with $x\not \in \bar{S}$ so that $x \in \rho(S \cup \bar{S})$.

In words, if the point $x$ is part of the algorithm input, it is always possible to send an update to force the point $x$ to be an output of the algorithm. It is interesting to compare this requirement with axioms of multi-winner social choice functions, as detailed e.g. in \cite{elkind2017properties}.

\end{definition}

\begin{theorem}
A set-choice algorithm with forceable winners is not $\ell$-NCC-vulnerable* for any $\ell$.
\end{theorem}

We prove the theorem using the two following claims.

\begin{claim} A strategy $s_j$ that satisfies condition $(i*)$ for a set-choice algorithm is explicitly-lying. 
\end{claim}
\begin{proof}
Consider a nature-input where agent $j$ receives no factual updates. To satisfy condition $(i*)$, it must send some ledger update for the algorithm output under $s_j$ to differ from that under $truth_j$. Since the union of all its factual updates is an empty set, it must hold that it sends a data point that does not exist there.
\end{proof}

\begin{claim}
\label{clm:condition2_forceable_domain}
An explicitly-lying strategy $s_j$ for a set-choice algorithm with forceable winners violates condition $(ii)$. 
\end{claim}

\begin{proof}
Consider the shortest nature-input $\ninput$ (in terms of number of elements) where $s_j$ sends a ledger update with an explicit lie $x$, and let $L_R = L_j(R_{\ninput, s_j}), F_R = F_j(R_{\ninput, s_j})$ be the union of all ledger, factual updates respectively by $j$. Let $S = F_R \cup L_R$, and $<i, \bar{S}>$ the nature-input element that generates a factual update of an agent $i \neq j$ that forces $x \in \rho(S \cup \bar{S})$ (such an element exist by the forceable winners condition). Let $E_1 = F_R \cup L_R \cup \bar{S}, E_2 = E_1 \setminus \{x\}$. Notice that $x \not \in \bar{S}$ (as required in Definition~\ref{def:forceable_winners} of forceable winners), but $x\in L_R$, and so $E_1 \neq E_2$. Also note that $x\not \in F_R$ (as it is an explicit lie). 
Let $\ninput_1, \ninput_2$ be $\ninput$ with an additional last element $E_1, E_2$ respectively. 

Now notice that $O_j(R_{\ninput_1,s_j}), O_j(R_{\ninput_2,s_j})$ are composed of the observed history $O_j(R_{\ninput,s_j})$, together with the observations following each of their different last elements. 
 As the last element is a factual update of an agent $i \neq j$, the agent sends a truthful ledger update. We then have $L_R \cup E_2 = L_R \cup ((F_R \cup L_R \cup \bar{S}) \setminus \{x\}) = (L_R \cup \{x\})\cup ((F_R \cup L_R \cup \bar{S}) \setminus \{x\}) = L_R \cup (F_R \cup L_R \cup \bar{S}) = L_R \cup E_1.$
Thus, the immediate algorithm output, and any further algorithm output following some ledger update by agent $j$ is taken over the same set, whether it is under $\ninput_1$ or $\ninput_2$, and so identifies. We conclude that $O_j(R_{\ninput_1,s_j}) =  O_j(R_{\ninput_2,s_j})$. 

On the other hand, the last algorithm output in $O_j(R_{\ninput_1,truth_j})$ is $\rho(F_R \cup E_1) = \rho(F_R \cup (F_R \cup L_R \cup \bar{S})) = \rho(S \cup \bar{S})$, and thus has the element $x$ by Definition~\ref{def:forceable_winners}. On the other hand, the last algorithm output in $O_j(R_{\ninput_2,truth_j})$ is $\rho(F_R \cup E_2) = \rho(F_R \cup ((F_R \cup L_R \cup \bar{S}) \setminus \{x\})) = \rho((F_R \cup L_R \cup \bar{S}) \setminus \{x\})$. Since $\rho$ is a set-choice algorithm, it does not output $x$ since it does not appear in the input set. 
\end{proof}

\begin{figure}
    \includegraphics[scale=0.6]{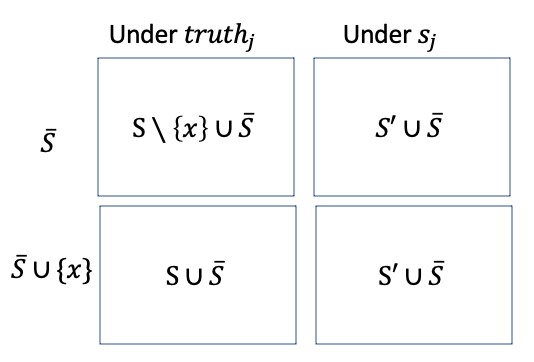}
    \caption{Demonstration of the proof of Claim~\ref{clm:condition2_forceable_domain}. $x$ is an explicit lie by agent $j$. $S$ is the state of the ledger under $truth_j$. $S'$ is the state of the ledger under $s_j$. $\bar{S}$ is a complementary set to $S$ from Definition~\ref{def:forceable_winners} (forceable winners). 
    Given that the next ledger update by a truthful agent is either $\bar{S}$ or $\bar{S} \cup \{x\}$ (which is represented by the rows), then the behavior under the different strategies (represented by the columns) is such that under $s_j$, the two underlying states of the world are the same, but not so under $truth_j$.}
    \label{fig:forceable}
\end{figure}

\begin{corollary}
\label{cor:k_center}
$k$-center is not $\ell$-NCC-vulnerable* for any $\ell$. 
\end{corollary}
\begin{proof}
$k$-center is a set-choice algorithm. We show that it has forceable winners. We show the construction for $R$, but the general $R^d, L_p$ is similar. Let some $S \subseteq R$ with $x \in S$. Let $\Delta = \max \{\max_{s\in S} |x-s|, 1\}$. Let $\bar{S} = \{ x + \Delta, x-\Delta\} \cup \{x + 10\Delta, \ldots, x + 10^{k-1} \Delta\}$. It must hold that $\rho(S \cup \bar{S}) = \{x, x + 10\Delta, \ldots, x+10^{k-1}\Delta\}$. 
\end{proof}

\begin{restatable}[]{corollary}{kmedian}
\label{cor:k_median}
$k$-median is not $\ell$-NCC-vulnerable* for any $\ell$.
\end{restatable}
The proof is given in Appendix~\ref{app:kcenter}. 

\section{Linear Regression under Continuous Communication}
\label{sec:lr}

In this section, we study the vulnerability(*) of linear regression.

\begin{definition}

Multiple linear regression in $d$ features $d-LR$: Given a set of data points $S$ with $n$ points, where the data points features are a $(d + 1)\times n$ matrix $\mathbf{X}$ with all elements of the first column normalized to 1, the targets are a $1\times n$ vector $\mathbf{y}$, then \[
\begin{split}
& \rho_{d-LR}(\mathbf{U^t}) = \\ & \rho_{d-LR}(\cup_{i=1}^t U_i) = \begin{cases} 
(\mathbf{X}^T\mathbf{X})^{-1} \mathbf{X}^T \mathbf{y} & \mathbf{X}\text{ columns are linearly independent} \\
Null & Otherwise. \\
\end{cases}
\end{split}
\]

We slightly abuse notation by defining $\rho_{d-LR}$ both as a function on a series of updates $\mathbf{U^t}$, as well as on a set of data points. 
The latter satisfies, as long as the columns are linearly independent, $\rho_{d-LR}(S) = \arg \min_{\mathbf{\beta} \in \mathcal{R}^d } \sum_{i\in |S|} (y_i - \sum_{j=1}^d x^j_i \beta_j)^2 $. We subsequently assume for simplicity that the columns are always linearly independent (e.g., by having a first ledger update with $d$ linearly independent features. The property is then automatically maintained with any future updates). 

\end{definition}

It is not difficult to find omission sneak attacks for linear regression, as we demonstrate in Figure~\ref{fig:sneak_lr}. 
\begin{figure}
    \hspace{-1cm}\includegraphics[scale=0.6]{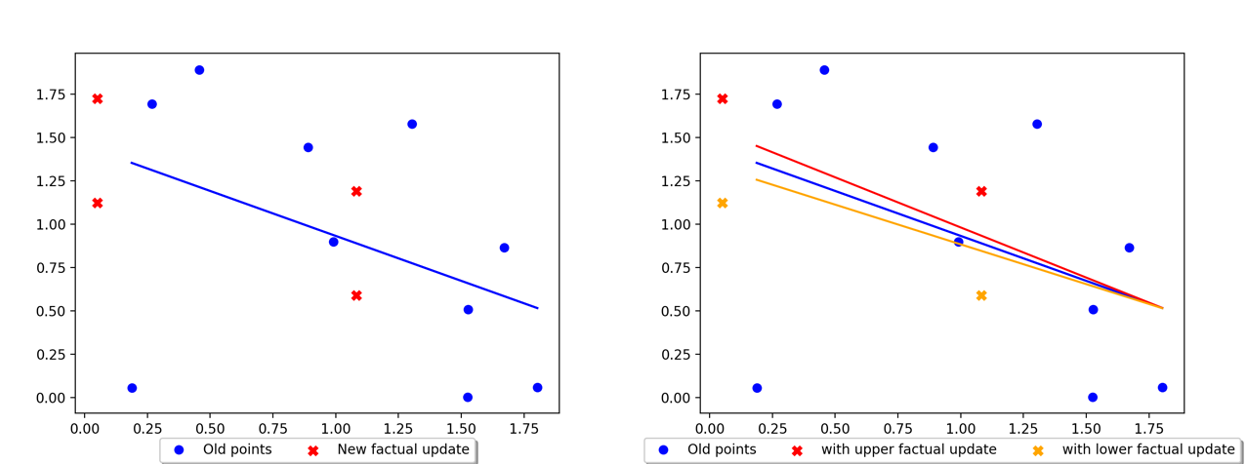}
    \caption{A sneak attack for simple linear regression. Since the points by others and the factual update of the agent yield the same LR estimator $\hat{\rho}$, the result of running the regression on all points is $\hat{\rho}$ regardless of what are the actual points by others.}
    \label{fig:sneak_lr}
\end{figure}
In Example~\ref{ex:LR_sneak_attack} in Appendix~\ref{app:lr}, we show a more complicated \textit{explicitly-lying} sneak attack for $1-LR$ (also called ``simple linear regression''). The attack can be generalized for $d-LR$. This yields

\begin{theorem}
$d$-LR is $1$-NCC-vulnerable. 
\end{theorem}

\subsection{Triangulation Attacks and Vulnerability*}

To study vulnerability*, we now define a stronger type of attacks and show they exist for $d-LR$, as long as $\ell \geq d + 2$. We name this type of attacks \emph{triangulation attacks}, and present a template parameterized by functions $f_1,...,f_{\ell-1}, h$ in Strategy Template~\ref{alg:triangulation_attack_template}. 

\begin{figure}
    \includegraphics[scale=0.4]{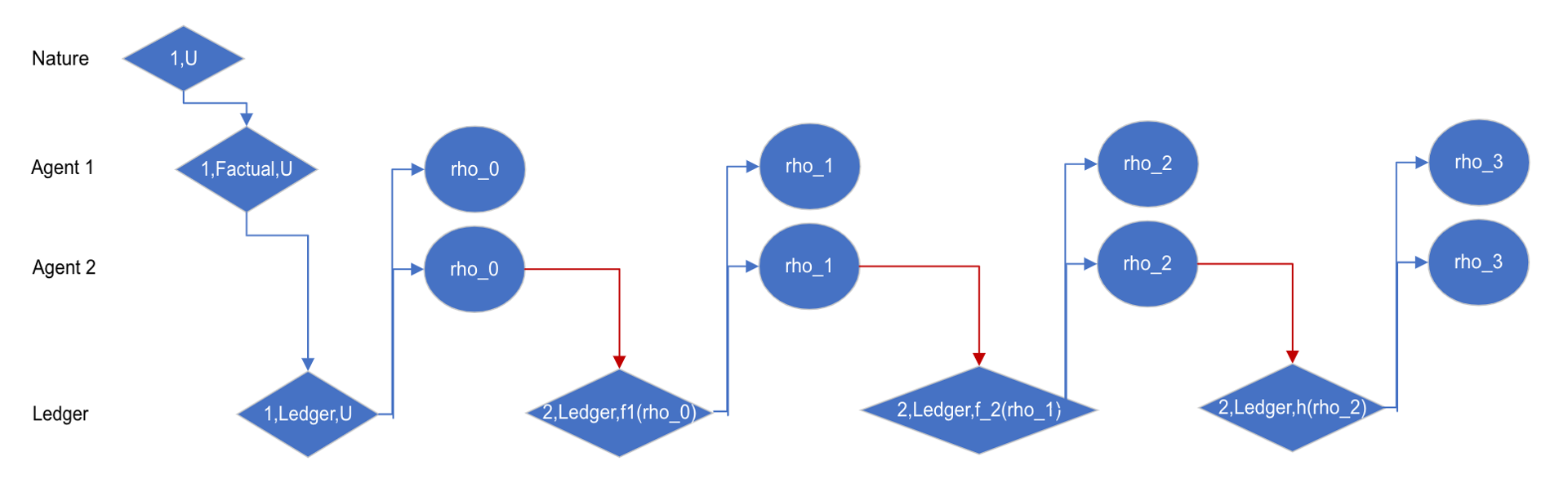}
    \caption{A general template for the triangulation attack, with $k=3$. Until the special conditions are met, and after the re-sync is done, the strategy behaves as $truth_j$.}
    \label{fig:triangulation_attack_template}
\end{figure}

\begin{strategytemplate}
\SetAlgoLined
\SetNoFillComment
\DontPrintSemicolon
\KwIn{Observed history $O_j$. Functions $f_1,\ldots, f_{\ell-1}, h$}
\KwOut{A ledger update $\hat{U}$}

Let $i=1$ if there is a factual update after the last ledger update by $j$. 

Otherwise, if a triangulation attack is ongoing, let $2 \leq i \leq \ell$ be its current step or else exit.

Let $\rho_{i-1}$ be the last algorithm output in $O_j$. 

\If{$1\leq i \leq \ell - 1$} {
    Return $<j, Ledger, f_i(\rho_{i-1})>$
} 

\ElseIf{$i = \ell$} {
    Return $<j, Ledger, h(\rho_{\ell-1})>$
}
 
 \caption{A template for a \emph{triangulation attack}}
 \label{alg:triangulation_attack_template}
\end{strategytemplate}

The idea of triangulation attacks is that for any state of the ledger, the attacker can find $\ell$ subsequent updates so that it can both infer the algorithm output if it applied strategy $truth_j$ instead of $s_j$ (using $f_1,...,f_{\ell-1}$ the ``triangulations''), and mislead others by the final update $h$. Informally, this attack has the desirable property that regardless of the state of the ledger (and how corrupted it may be by previous updates of the attacker), the attacker can infer the true state. 

As in the case of the sneak attack, we should show the strategy template can be implemented using only the information in $O_j$.

\begin{restatable}[]{lemma}{triangulationAttackImp}
The triangulation attack is well defined, i.e., the conditions in lines $1$ and $2$ can be implemented using only information available in $O_j$. The 
assignment in line $3$ is valid, that is, given that line $3$ is executed there exists an algorithm output in $O_j$. 
\end{restatable}

We defer the proof details to Appendix~\ref{app:st_lemmas}. 

We now prove there is a triangulation attack for $d-LR$ with $\ell \geq d+2$. 

\begin{theorem}
\label{thm:d_lr_upper_bound}
$d-LR$ is $(d+2)$-NCC-vulnerable* using a triangulation attack $f_1,...,f_{d+1},h$. 
\end{theorem}
\begin{proof}

We shortly outline the overall flow of the proof. First, we give explicit construction of the $\{f_i\}_{1\leq i \leq d+1}$ functions. This suffices to show that condition $(ii)$ is satisfied, which means there is an inference function $i(O_j)$ that maps observed histories under $s_j$ to the last algorithm output under $truth_j$. Given that inference function, we construct $h$ and show that with it condition $(i*)$ is satisfied. We give a formal treatment of inference function in Definition~\ref{def:inference_function} and Lemma~\ref{lem:inference_function} of Appendix~\ref{app:max_illustration}, but for our purpose in this proof it suffices that it is a map as specified.  

\textbf{Construction of $\{f_i\}_{1\leq i \leq d+1}$ and condition $(ii)$:}

Let $$\rho_{i-1} =  \begin{bmatrix} \rho_{i-1}^1 \\ \ldots \\ \rho_{i-1}^{d+1} \end{bmatrix}$$
be the last algorithm output before the application of $f_i$. Define

$f_i(\rho_{i-1}) = (X_i, y_i)$,

where $X_i$ is the $(d+1) \times 1$ vector with $X_i^1 = X_i^i = 1$, and $$y_i = \begin{cases}\rho_{i-1}^1 + 1 & i = 1 \\ \rho_{i-1}^1 + \rho_{i-1}^i + 1 & 2 \leq i \leq d+1\end{cases}.$$ 

Let $R_{\ninput, s_j}$ be a run with some nature-input $\ninput$ and $s_j$ the triangulation attack with the specified $f_1,\ldots, f_{d+1}$ (and any function $h$).
Consider all the factual updates by agents $\neq j$ induced by $\ninput$. They are each of the form of $(X',y')$, where $X'$ is of size $n \times (d+1)$ and $y'$ is $n \times 1$, and where $n$ is the number of data points in the update. To consider all factual updates of the agents $\neq j$, we can vertically concatenate these matrices. Let this aggregate be denoted $X_{F,-j}, y_{F,-j}$. 
Similarly, let $X_{F,j}, y_{F,j}$ be the concatenation of all factual updates by $j$ . Let the concatenation of all ledger updates by $j$ before submission of any of the $f_i$ updates be $X_{L,j}, y_{L,j}$. Recall that we denote by $\rho_0, \ldots, \rho_{d+1}$ the algorithm outputs (right before, and after each $f_i$, e.g. $f_1$ is applied after $\rho_0$ and generates $\rho_1$).
Let $X'_i, y'_i$ be the (concatenated) inputs to the $d-LR$ algorithm that generate $\rho_i$. In terms of the defined variables above, we can write:

\begin{equation}
\label{eq:triangulation_linear}
\begin{split}
& (X'_i)^T X'_i =  (X_{F,-j})^T X_{F,-j} + (X_{L,j})^T X_{L,j} + \sum_{t=1}^i (X_i)^T X_i, \\
& (X'_i)^T y'_i =  (X_{F,-j})^T y_{F,-j} + (X_{L,j})^T y_{L,j} + \sum_{t=1}^i (X_i)^T y_i,
\end{split}
\end{equation}

To show that condition $(ii)$ holds, it suffices to show that we can infer the last algorithm output $\rho_{truth}$ of the run $R_{\ninput, truth_j}$. Let $(X_F, y_F)$ be the concatenation of all factual updates of all agents, then it is the input that generates $\rho_{truth}$, and it holds that:

\begin{equation}
\label{eq:factual_linear}
\begin{split} & (X_F)^T X_F =  (X_{F,-j})^T X_{F,-j} + (X_{F,j})^T X_{F,j} \\
& (X_F)^T y_F =  (X_{F,-j})^T y_{F,-j} + (X_{F,j})^T y_{F,j}
\end{split}
\end{equation}
 
Since in Equation~\ref{eq:factual_linear}, besides $X_{F,-j}, y_{F,-j}$, all RHS variables are observed history under $s_j$, we conclude that it is enough to deduce $X_{F,-j}^TX_{F,-j}, X_{F,-j}^Ty_{F,-j}$ in order to infer $(X_F)^T X_F, (X_F)^T y_F$, and thus also the last algorithm output under $truth_j$ which is 
$((X_F)^T X_F)^{-1} (X_F)^T y_F$.

Let $(X_{F,-j})^T X_{F,-j} \stackrel{def}{=} \begin{bmatrix} \Sigma_{1,1} & \ldots & \Sigma_{1,d+1} \\ \ldots \\ \Sigma_{d+1,1} & \ldots & \Sigma_{d+1,d+1} \end{bmatrix}, (X_{F,-j})^T y_{F,-j} = \begin{bmatrix} \sigma_1 \\ \ldots \\ \sigma_{d+1} \end{bmatrix}$.

For every $0\leq i\leq d+1$, we have \begin{equation}
\label{eq:output_linear}
    (X'_i)^T X'_i \rho_i = (X'_i)^T y'_i.
\end{equation} 

By the construction of $f_i$, we can rewrite these equations in the following way. Let $D^i$ be the $(d+1)\times(d+1)$ matrix with $D^i_{1,1} = D^i_{i,1} = D^i_{1,i} = D^i_{i,i} = 1$, and all other elements zero. Let $v^i$ be the $1\times(d+1)$ vector with 
$$v^i_1 = v^i_i = \begin{cases}
\rho_0^1 + 1 & i = 1 \\
\rho_{i-1}^1 + \rho_{i-1}^{i} + 1 & i > 1
\end{cases},$$
and all other elements zero. 

We have for $0 \leq i \leq d+1$:

\begin{equation}
\label{eq:output_linear_detailed}
    (\begin{bmatrix} \Sigma_{1,1} & \ldots & \Sigma_{1,d+1} \\
    \ldots \\
    \Sigma_{d+1,1} & \ldots & \Sigma_{d+1,d+1} 
    \end{bmatrix} +  \sum_{t=1}^i D^i)\rho_i = \begin{bmatrix} \sigma_1 \\ \ldots \\ \sigma_{d+1}\end{bmatrix} + \sum_{t=1}^i v^i.
\end{equation} 

If we examine the differences between the $i$ equation and the $i-1$ equation, we get for $1 \leq i \leq d+1$, 

\begin{equation}
\label{eq:output_linear_diff}
    (\begin{bmatrix} \Sigma_{1,1} & \ldots & \Sigma_{1,d+1} \\
    \ldots \\
    \Sigma_{d+1,1} & \ldots & \Sigma_{d+1,d+1} 
    \end{bmatrix} + \sum_{t=1}^{i-1}D^t )(\rho_i - \rho_{i-1}) = v^i - D^i \rho_i.
\end{equation} 

Notice that for any $1\leq i \leq d+1$, $v^i - D^i\rho_i $ is not the zero vector. If it was, since $(X_{F,-j})^TX_{F,-j} + \sum_{t=1}^{i-1}D^t$ is invertible, we will have that $\rho_i = \rho_{i-1}$, which would contradict the following claim:

\begin{restatable}[]{claim2}{dlrPointLine}
\label{clm:d_lr_point_line}

For every algorithm output $\rho = \begin{bmatrix} \alpha_1 \\ \ldots \\ \alpha_{d+1}\end{bmatrix}$, and a single point update $U* = (X^* = \begin{bmatrix} 1 & x_1 & \ldots & x_d \end{bmatrix}, y^*)$ so that $X^* \cdot \rho \neq y^*$, the new algorithm output $\rho'$ for the data with $U*$ satisfies $\rho' \neq \rho$, and has a different value at $X^*$ than $X^* \cdot \rho$. 
\end{restatable}

The proof of the claim is given in Appendix~\ref{app:lr}. 

Moreover, $v^i - D^i \rho_i$ by definition is a vector that has all elements $0$ besides element $1$ and $i$ that are $\rho_{i-1}^1 + \rho_{i-1}^i + 1 - \rho_i^1 - \rho_i^i \neq 0$ (since it is not a zero vector), and so the $i$-th element of the vector is non-zero. Therefore, for the vector $w^i \stackrel{def}{=} v^i - D^i\rho_i - \sum_{t=1}^{i-1} D^t (\rho_t - \rho_{t-1})$, the $i$-th element is non-zero as well (Since $\sum_{t=1}^{i-1} D^t \rho_i$ has all elements with index higher than $i-1$ as zero). For any $w^t$ with $t\leq i-1$, all elements with index higher than $i-1$ are zero. Therefore, the set $\{w^i\}_{1\leq i \leq d+1}$ is linearly independent, and the matrix $W$ where each column $i$ is $w^i$ is invertible. If we let $M_{\rho}$ be the matrix where each column $i$ is $\rho_i - \rho_{i-1}$, we can rewrite Eq~\ref{eq:output_linear_diff} as $(X_{F,-j})^TX_{F,-j} M_{\rho} W^{-1} = I$,
where $I$ is the $(d+1)\times (d+1)$ identity matrix. We conclude that $M_{\rho}$ is invertible and $(X_{F,-j})^TX_{F,-j} = W M_{\rho}^{-1}$. We can directly calculate the RHS of this expression from the observed history under $s_j$, and by the first equation of Eq~\ref{eq:output_linear_detailed} we can infer $(X_{F,-j})^T y_{F,-j} = (X_{F,-j})^TX_{F,-j} \rho_0$, overall concluding the proof for condition $(ii)$. 

\textbf{Construction of $h$ and condition (i*)}. Let $i$ be the inference function (which existence is guaranteed by the previous discussion) that matches observed histories running $s_j$ with the true algorithm outputs under $truth_j$. I.e., we has  $i(O_j) = \rho_{truth}$. Let the last algorithm output in $O_j$ be $\rho_{last} = \begin{bmatrix} \rho_{last}^1 \\ \ldots \\ \rho_{last}^{d+1} \end{bmatrix}$. Let $h(O_j) = \begin{cases}
(\begin{bmatrix} 1 & 0 & \ldots & 0 \end{bmatrix}, \begin{bmatrix} \rho_{last}^1 + 1 \end{bmatrix}) & i(O_j) = \rho_{last}\\
\text{No update} & Otherwise.
\end{cases}$. 

If $\rho_{last} \neq \rho_{truth}$, $h$ does not send an update, and so for the nature-input that has observed history $O_j$ the last algorithm output under $s_j$ is different than that under $truth_j$, as required by condition $(i*)$. 

 If $\rho_{last} = \rho_{truth}$, $h$ sends an update with a point $(X,y)$ that satisfies $X \cdot \rho_{last} = \rho_{last}^1 \neq \rho_{last}^1 + 1 = y$. By Claim~\ref{clm:d_lr_point_line}, the resulting algorithm output is different from $i(O_j)$. 

\end{proof}

We demonstrate the construction and inference of the triangulation attack in an open-source implementation \href{https://github.com/yotam-gafni/triangulation\_attack}{https://github.com/yotam-gafni/triangulation\_attack}. Figure~\ref{fig:tri_script} shows a run of the attack for a random example for $2$-LR.  

\begin{figure}
    \includegraphics[scale=0.6]{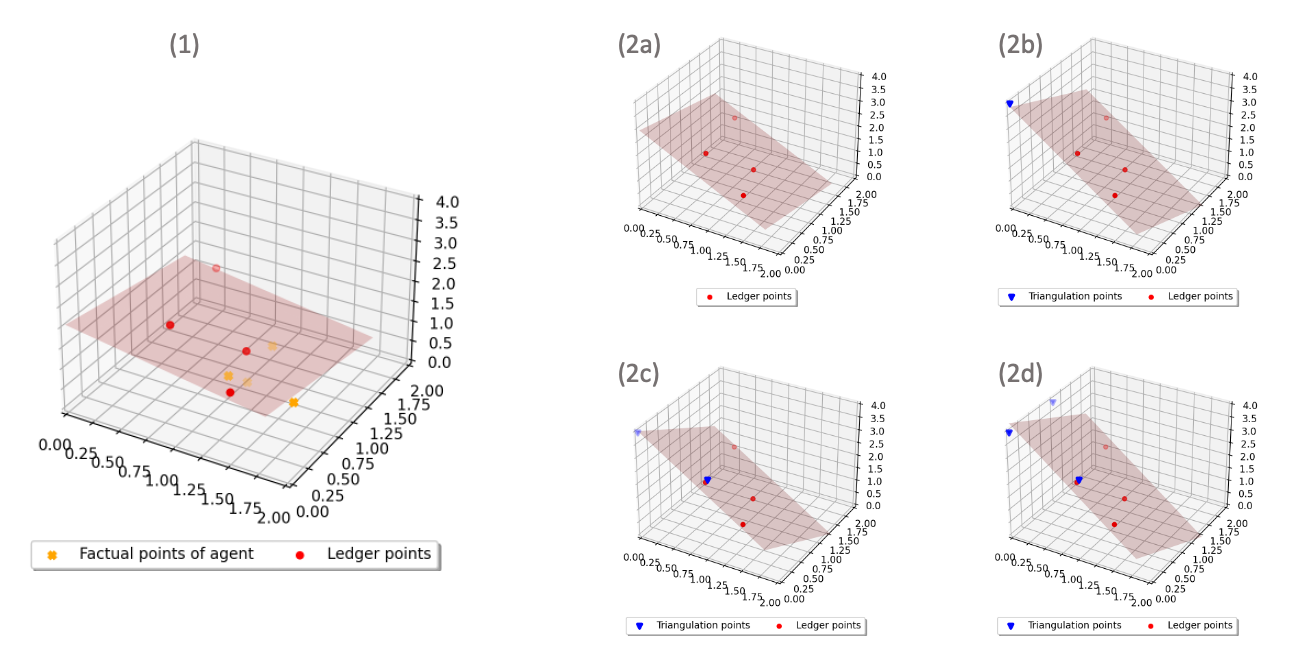}
    \caption{A script-run triangulation attack for $2$-LR. The round red points represent an existing state of the ledger. The yellow x points (in (1)) represent a new factual update for the strategic agent. The red line in (1) represents the resulting linear regression estimator, if the agent reports truthfully. The four figures (2a)-(2d) show the flow of our triangulation attack construction. In (2a) is the last state of the ledger before the triangulation, with no triangulation point sent by the strategic agent. The rest of (2b)-(2d) consecutively add triangulation points (blue triangles). At the end of the triangulation attack (after (2d)), the linear regression estimator is different than in (1). It is possible to infer the estimator in (1) using knowledge of the triangulation points and estimators of (2a)-(2d) (without knowledge of the red points).}
    \label{fig:tri_script}
\end{figure}

We show an asymptotically matching lower bound for triangulation attacks. 

\begin{theorem}
\label{thm:d_lr_lower_bound}
There is no triangulation attack for $d-LR$ with $d - 2$ or less functions (i.e., $\ell \leq d-2$).
\end{theorem}
\begin{proof}
Consider all nature-input elements that are of the form $<i,(X,y)>, <j,(\bar{X}_j,\bar{y}_j)>$, where $X$ is a $(d+1)\times(d+1)$ matrix, and $y$ is the $(d+1) \times 1$ zero vector. $(\bar{X}_j,\bar{y}_j)$ of the same sizes but without any restriction over $\bar{y}_j$. We show that for any triangulation attack $s_j$, we can find two nature-inputs among this family with different observed history under $truth_j$, but the same observed history under $s_j$. 

By the choice of $y$, the first algorithm output satisfies $\rho_0 = (X^TX)^{-1}X^T y = \textbf{0}$. As we know from the proof of Theorem~\ref{thm:d_lr_upper_bound}, in particular Equation~\ref{eq:output_linear_detailed} (where it was done for a specific given triangulation attack), that the attack generates $d-1$ vector equations for $X^T X$ (including the one over $\rho_0$). We also know that the first row of $X^T$ is all $1$ elements. We can make it a stricter constraint by demanding that the first row of $X^TX$ is of the form $\begin{bmatrix} d+1 & 0 & \ldots & 0\end{bmatrix}$. Then, the principal sub-matrix of $X^TX$ (removing the first row and column) is a general PSD matrix (as a principal submatrix of the $X^TX$ PSD matrix). To uniquely determine such a matrix of size $d\times d$, we need $d$ vector equations, but the triangulation equations only yield $d-1$ such equations. So there are some $X_1 \neq X_2$ that are in the family of nature-inputs and have the same observed history under $s_j$. Fix some invertible $\bar{X}_j$. Since $(X_2^TX_2 + \bar{X}_j^T \bar{X}_j) \neq (X_1^TX_1 + \bar{X}_j^T \bar{X}_j)$, there must be some $v$ so that

$$(X_2^TX_2 + \bar{X}_j^T \bar{X}_j)^{-1} v \neq (X_1^TX_1 + \bar{X}_j^T \bar{X}_j)^{-1} v.$$

If $\bar{X}_j \bar{y}_j = v$, then the last algorithm outputs under $truth_j$ are different for $X_1, X_2$, which holds choosing $\bar{y}_j = (\bar{X}_j)^{-1} v$.
\end{proof}

\section{The Periodic Communication Protocol}
\label{sec:alternate_models}
The periodic communication protocol simulates a system where update rounds are initiated by the system manager (or ledger), and not by the agents themselves. After each round, the ledger shares the algorithm output with all agents. 
The definitions of section 2 remain consistent with this periodic setting, with the following minor changes:
\begin{itemize}

\item Since all updates by different agents in a certain round are aggregated together, the distinction of $\ell$ subsequent updates becomes irrelevant and we omit it. 

\item An identifier of the round number $r$ is added to each nature-input element. That is, each element is $<j,U,r>$, with an agent $j\in \agents$, an update $U$, and a round number $r$.\footnote{Round numbers are assumed to have natural properties: They are monotonically increasing
with later elements of the nature-input series, each agent has at most one nature-input element assigned to it per round. The first round is $r=1$.}

\end{itemize}

\begin{figure}
    \includegraphics[scale=0.8]{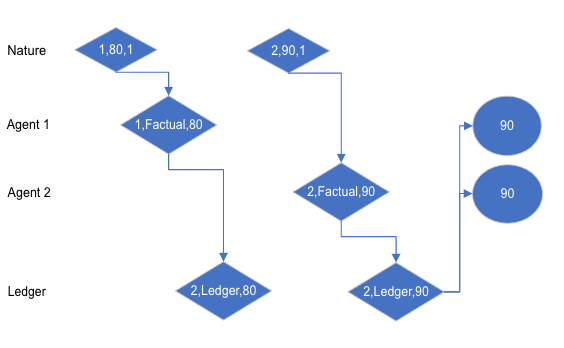}
    \caption{A periodic protocol run for $\ninput = (<1,90,1>, <2,90,1>)$. Both agents are truthful.}
    \label{fig:example_periodic}
\end{figure}

\begin{protocol}
\SetNoFillComment
\SetAlgoLined
\DontPrintSemicolon
\KwIn{Nature-input $\ninput$}
\KwOut{Full Messaging History}

Let $r_{max} = \max_{<j,U,r> \in \ninput} r$.

\tcc{For each round of updates}
\For {$\bar{r}:=1$ to $r_{max}$} {

    \tcc{For each update in round $\bar{r}$}
    \For{ Element $<j,U_{fact},r>$ with $r = \bar{r}$}
    {
        Nature sends a message to $j$ with $<j,Factual, U_{fact}>$;
    
    }
    \For {agent $i := 1$ to $n$} {
        \If{agent $i$ wishes to send a ledger update $U_{ledg}$} {
                $i$ sends a message to Ledger with $<i, Ledger, U_{ledg}>$;
        
        }
    }
    Ledger sends a message to all with $\rho$'s algorithm output over all the past ledger updates;  
}
 \caption{The periodic communication protocol}
 \label{proto:periodic}
\end{protocol}

We now show that indeed periodic communication is strictly less vulnerable to attacks, both for $k$-center and $d-LR$. 

\begin{theorem}
\label{thm:lr_periodic}
$d-LR$ is not NCC-vulnerable in the periodic communication protocol.
\end{theorem}

We prove this theorem using a more general lemma. We first define three useful properties of a minimization task:
\begin{definition}

A multi-set \emph{minimization problem} $C$ is of the form $\rho(S) = \arg \min_{\rho'} C(S,\rho')$, where $S$ is the algorithm input, $C$ is a cost function and $\rho'$ is some possible algorithm output.  

A minimization problem $C$ is \emph{separable} if $C(S_1 \uplus S_2, \rho) =  C(S_1, \rho) + C(S_2, \rho)$. Separable minimization problems are also homogeneous in the sense that: $C(S \times \lambda, \rho) = \lambda C(S, \rho)$. 

A minimization problem has \emph{a unique solution} if for every input $S$ it has a single algorithm output $\rho'$ that attains the optimal goal.

A minimization problem is \emph{non-negative} if for every input $S$ and possible algorithm output $\rho'$, $C(S, \rho') \geq 0$. 
\end{definition}

We know that $d-LR$ under the restriction mentioned (independent columns) has a unique solution. It is also immediate from its definition as an optimization problem that it satisfies separability. 
Theorem~\ref{thm:lr_periodic} now follows on the following general lemma:

\begin{figure}
    \includegraphics[scale=0.6]{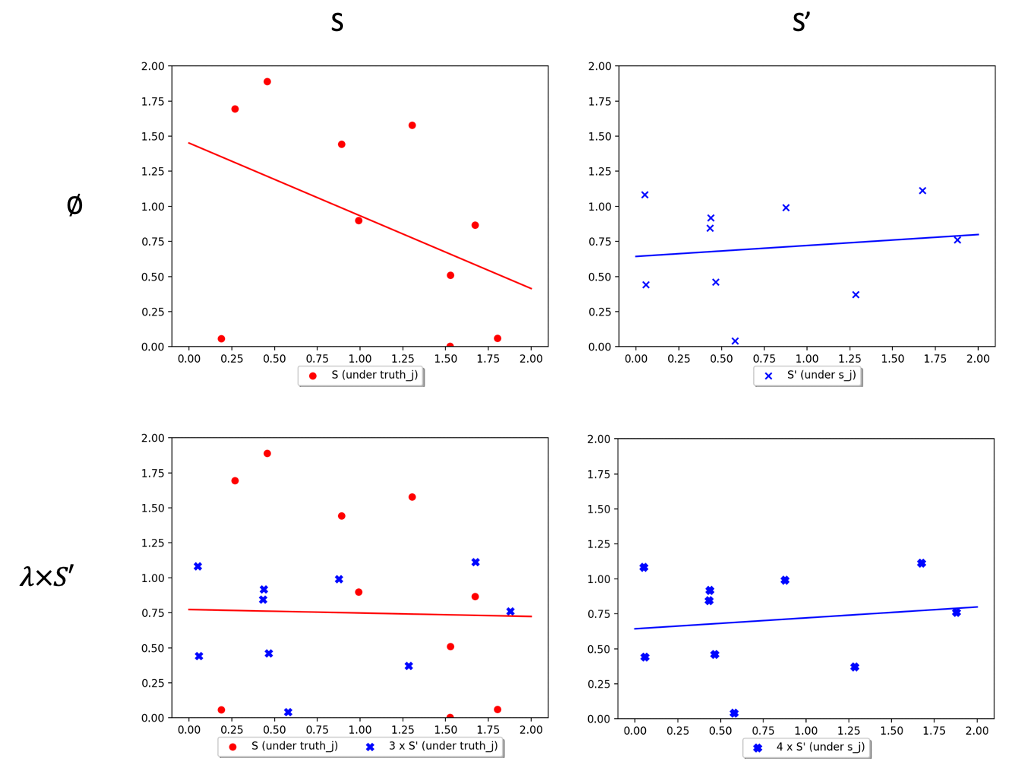}
    \caption{
    Demonstration of the proof of Lemma~\ref{lem:periodic_minproblem}. Under $s_j$, the estimator for $1$-LR is the same whether the other agent \textit{additionally} submits $\emptyset$ or $3 \times S'$, but not so under $truth_j$.}
    \label{fig:periodic_min}
\end{figure}

\begin{lemma}
\label{lem:periodic_minproblem}
Any multi-set algorithm $\rho$ that can be formalized as a minimization problem with separable, non-negative minimization goal $C$ with a unique solution is not NCC-vulnerable in the periodic communication protocol.
\end{lemma}
\begin{proof}
Assume the algorithm is NCC-vulnerable in periodic communication with some strategy $s_j$. By condition $(i)$, there is nature input $\ninput$ so that the last algorithm output under $truth_j$ is $\rho$ and under $s_j$ is $\rho'$. Let $S, S'$ be some underlying input to generate $\rho, \rho'$ respectively. Since $\rho, \rho'$ are unique solutions, it must hold that $0 \leq C(S', \rho') < C(S', \rho), 0 \leq C(S, \rho) < C(S, \rho')$. Let $\Delta = C(S, \rho') - C(S, \rho), \delta = C(S',\rho) - C(S', \rho'), \lambda = \lceil \frac{\Delta}{\delta} \rceil + 1$. Now assume that some agent $\neq j$ sends $S' \times \lambda$ in the last round of $\ninput$ (call this extension $\ninput'$. If all agents already send an update in this round, add $S' \times \lambda$ to one of these agents' update). Under $s_j$, we have that $C(S' + (S' \times \lambda), \hat{\rho}) = (\lambda + 1)C(S', \hat{\rho})$ and so $\rho$ remains the unique solution (The argmin does not change under multiplication of the cost function). Under $truth_j$, we have 
\[
\begin{split}
& C(S + (S' \times \lambda), \rho') = C(S, \rho') + \lambda C(S', \rho') = C(S, \rho) + \Delta + \lambda(C(S',\rho) - \delta) = \\
& C(S \cup (S' \times \lambda), \rho) + \Delta - \lambda \delta < C(S \cup (S' \times \lambda), \rho),
\end{split}
\]
and so $\rho$ is not the optimal algorithm output. 

We thus have a violation of condition $(ii)$: There are two nature inputs ($\ninput, \ninput'$) with the same observed history under $s_j$ but different under $truth_j$. 
\end{proof}

In the appendix, we prove a similar result for $k$-center. The result also holds for $k$-median and is done by extending the construction of Corollary~\ref{cor:k_median}. 

\begin{restatable}[]{theorem}{kcenterPeriodic}
\label{thm:kcenter_periodic}
$k$-center is not vulnerable under the periodic communication protocol. 
\end{restatable}

\section{Discussion}
\label{sec:discussion}

In this work, we lay the groundwork for the study of exclusivity attacks in long-term data sharing. We present two protocols for long-term communication and show that the choice of protocol, as well as the number of Sybil identities an attacker may control, matters for the safety of the system. We do so by analyzing two representative and popular algorithms of supervised and unsupervised learning, namely linear regression and k-center. We show that the distinction between omission and explicitly-lying attacks has theoretical significance, and present two general attack templates that are useful to consider against any possible algorithm. However, we believe that these are the first steps and that there is much more to study regarding systems’ safety from exclusivity attacks. We now expand on a few possible future directions. 

\subsection{Further Model Extensions}

\subsubsection{Varying Temporal Resilience}

In our model, condition $(ii)$ requires one pair of confounding nature-inputs, i.e., one state of the world where the agent can not infer the true best model fit. However, when dealing with collaborative computing, some organizations may have different ``temporal resilience''. While some depend daily on the learned parameters, others operate in longer time scales such as issuing weekly or monthly reports. In such cases, an attacker $j$ may be willing to incur being confounded, as long as the confusion is bounded within a small number of algorithm outputs, after which it can again infer the true parameters. Adjusting the model to accommodate such heterogeneous preferences and how they affect the results can be interesting. 

\subsubsection{Horizontal vs. Vertical Data Split}

In multi-agent collaborative learning tasks, a common distinction is between ``Horizontal'' and ``Vertical'' data split \cite{yang2019federated}. A horizontal split is when the set of features is shared among agents, but the data points may differ. Vertical split is when the data points are related to the same users, but the feature space is different among agents. While our model is general and can accommodate both cases, our results largely deal with the horizontal case, and it would be interesting to look into the vertical case as well. 

\subsubsection{Application to Silo-ed Federated Learning}

A leading motivation for developing the theory in this work is to apply it to federated learning, in particular in the context where the contributors are a few large firms (referred to as Silo-ed federated learning in \cite{kairouz2019advances}). As we know from the case of the Average algorithm \cite{ncc4ml}, changing the amount of information shared with the agents can determine the safety of the collaboration (In the Average case, whether the denominator of the number of samples is shared alongside the average itself). 
Applied in the context of federated learning, design choices such as split learning \cite{gupta2018split}, keeping hyper-parameters at the aggregator level and not the client level (Notice that this is in contrast with the design of the popular FederatedAveraging algorithm \cite{federatedAveraging}!), or varying the accuracy of the model supplied to agents \cite{modelDegradationFederating}, can be promising ideas to deter NCC attacks. Another issue that needs to be addressed is that of learning being resistant to permutations over the order of samples \cite{ravanbakhsh2016deepPermutation}. In the set and multi-set algorithms we treat in this work, the order of the updates does not matter for the algorithm output, and so it is possible to strategically control how and when to share factual data, for example in sneak attacks. However, in training neural nets, the order of feeding samples can change the final model (See the discussion in 1.4.2 in \cite{montavon2012neural}). 

\subsubsection{Relaxing the NCC Requirements and Approximate Mechanisms}
The requirement from exclusivity attacks to be able to infer the \emph{exact} true algorithm output seems harsh. This is especially true when dealing with statistical estimators, that by their nature are prone to noise. So, it is interesting to see how do the positive results of our work (in the sense of no-vulnerability of an algorithm under some settings) hold when attackers are willing to suffer some $\epsilon$ degradation of the algorithm output in comparison with the true result (under some appropriate metric). Such a discussion also opens the gate to a mechanism design problem. Once agents are willing to suffer some degradation of the model, it is possible to consider approximate algorithms that have better incentive-compatible properties than the standard algorithm. However, simply adding noise to an algorithm does not guarantee that it is safer. For example, consider that we take the one-shot sum algorithm and add some $0$-mean noise with expected variance $\epsilon$. Under $truth_j$, agent $j$ will have a difference of $\epsilon$ from the true sum in expectation. If agent $j$ attacks by adding $\delta$ to its true number in the ledger update, and then reduces $\delta$ from the algorithm output, its expected deviation from the true sum remains $\epsilon$, but it is able (by choosing $\delta$ right) to mislead others on average by more than $\epsilon$. Therefore, we remark that a good approximate algorithm to deter attacks should somehow guarantee that the attack process amplifies the error to hurt the attacker. 

Another interesting option that is possible once dealing with a relaxation of NCC is to have different algorithm outputs sent to different agents, i.e., the protocol does not share a global algorithm output each time with all agents, but gives a different response to each, hopefully in a way that helps enforce incentive compatibility. 

\section*{Acknowledgements}

Yotam Gafni and Moshe Tennenholtz were supported by the European Research Council (ERC) under the European Union’s Horizon 2020 research and innovation programme (Grant No. 740435). 

\bibliographystyle{ACM-Reference-Format}
\bibliography{refer}

\appendix

\section{Related Work}
\label{app:related_work}

\subsection{Collaborative Machine Learning}

Data sharing between companies and institutions is an emerging phenomenon in the data economy \cite{data_sharing_europe}, still under-performing its full potential. Many companies, cloud services, and government initiatives \cite{gaia-x} offer frameworks and APIs to facilitate such exchange, as well as some decentralized blockchain services \cite{ocean_protocol}. However, the current focus of these services is in organizing the nuts and bolts of such procedures (e.g. in terms of software, scale, and cyber-security), and not in ensuring incentive-compatibility, in particular dealing with exclusivity attacks. 

Mechanisms based on VCG and the Shapley-value were suggested as a method to construct general incentive-compatible mechanisms for data collaboration in  \cite{IC_distributed_classification}.  We highlight three main aspects of that work that differ from our approach: They assume the existence of a test set for each agent to compare other firms' inputs (separate from the data set it communicates with others); They consider a one-shot process rather than a continuous one, and they use monetary transfers while we consider data sharing a barter between firms without exchanging money. The assumptions we share with this work are that the true output of the machine learning algorithm is the best parameter possible to learn and that the agents' utilities are the NCC framework utilities. 

In \cite{collab_ai_blockchain} the authors consider continuous data sharing implemented by a blockchain, with various incentive mechanisms depending on the assumptions for agents' incentives. An essential difference with our work is that the data is assumed to be posted publicly and is thus known to all agents. This is an issue both by itself in terms of privacy, but also when designing incentives. As we will see, the uncertainty regarding other agents' data is essential for the safety of certain mechanisms under the NCC assumptions.

Federated learning is a popular framework for decentralized machine learning with private information \cite{federatedAveraging}. The general scheme has each agent perform stochastic gradient descent (SGD) by itself and share the gradients with an aggregator, in order to train a global model. The global model is public, and this is inherent to the operation of the mechanism since the agents are expected to calculate the gradients.
There is a natural free-rider attack (mentioned in \cite{kairouz2019advances}) in such cases where the agent shares no data (or, possibly, a small amount of the data it has) and later completes the training locally based on the global model and its remaining private data. This attack form fits within our framework of exclusivity attacks, and we discuss in Section~\ref{sec:discussion} how our insights may apply to it.  

There is a line of work that is orthogonal to ours \cite{workersIncent,crowdBC}, which focuses on assigning model training tasks to workers, in order to offload computation from being done by the central authority, or on-chain in the case of a decentralized blockchain. We note that mechanisms built for this task are different in nature and purpose from data sharing mechanisms. 

\subsection{Linear Regression and \texorpdfstring{$k$}---Center in Adversarial Settings}

In this work, we use linear regression and the $k$-Center and $k$-Median problems to examine our NCC utilities framework. 

Linear regression \cite{lr_book} is a well-known regression mechanism. We study Multiple Linear Regression with $d$ features.

In \cite{LRprivacy} the authors study linear regression with users that have privacy concerns.  In\cite{cai2015optimum} the authors suggest a mechanism using optimal monetary transfers to induce statistical estimation using reports by workers that exert effort to attain more precise estimations. The mechanism is shown to generalize to more general classes of regression than linear regression.  Following the framework of ``Dueling algorithms'' in \cite{immorlica2011dueling}, in \cite{ben2019regression} the authors consider firms optimizing their regression models to better satisfy a subset of the users relative to the opponent. In \cite{lr_strategic_sources} the authors consider firms that control the level of noise they add to the dependent variable, and aim to balance between privacy (more noise) and model accuracy (less). In  \cite{dekel2010incentive} the authors consider a general regression learning model where experts have strong opinions and wish to influence the resulting model in their favor. As one can see, there are many strategic reasons to manipulate regression tasks, but the NCC setting is a significant and understudied one. 

The $k$-center and $k$-median problems \cite{hakimi1964optimum,hochbaum1985best} are associated with clustering or facility location algorithms. Facility location problems were studied extensively in strategic settings \cite{facility_location_survey} \cite{moulin1980strategy}. The main focus is usually on strategic users, that may manipulate reporting of their location to influence the facility locations' outcome \cite{procaccia2013approximate}. For this purpose, strategy-proof mechanisms are developed, with the goal of a small approximation ratio relative to the optimal (without strategic consideration) algorithm. Our setting is different as we consider firms that acquired knowledge of users' preferences (or locations), and their goal of manipulation is not to benefit the users they have information about but to know the resulting aggregate outcome better than the other firms.

\section{Illustration of Preliminaries Using the Max Algorithm}
\label{app:max_illustration}

We define the max algorithm:

\begin{restatable}[]{definition}{algMax}

Each update is a real number. $\rho_{max}(\mathbf{U^t}) = \max_{1\leq i \leq t} U_i$. 

\end{restatable}

\begin{proposition}
\label{prop:max_vulnerable*}
max is not $\ell$-NCC-vulnerable* for any $\ell$.
\end{proposition}

\begin{proof}
Consider w.l.o.g. agent $1$ has a strategy $s_1$ that satisfies conditions $(i*)$ and $(ii)$. For the nature-input $\mathcal{I} = (<2,90>)$, under $truth_1$ agent $2$ receives a factual update $<2,Factual,90>$ and then updates with $<2,Ledger,90>$, resulting in algorithm output $90$. By condition $(i*)$, under $s_1$ the algorithm output after the full run must differ from $90$. Agent $1$ must thus update with at least one update of the form $<1,Ledger, x>$ with $x$ larger than $90$. Now consider the two nature-inputs $\mathcal{I}' = (<2,90>, <2,\frac{1}{3} \cdot x + \frac{2}{3} \cdot 90>), \mathcal{I}'' = (<2,90>, <2,\frac{2}{3} \cdot x + \frac{1}{3} \cdot 90>)$. Since $x \neq 90$, the observed histories under $truth_1$ for $\ninput', \ninput''$ are not the same, as the last algorithm output are  $\frac{1}{3} \cdot x + \frac{2}{3} \cdot 90 \neq \frac{2}{3} \cdot x + \frac{1}{3} \cdot 90$ respectively. I.e., 
\begin{equation}
\label{eq:observed_truth_neq}
O_1(R_{\ninput', truth_1}) \neq O_1(R_{\ninput'', truth_1}).
\end{equation}
Since the prefix of $\ninput', \ninput''$ is $\ninput$, we know that under $s_1$ by the end of the first round the algorithm output is $x$. After agent $2$ receives the second factual update and updates truthfully, the observed history (in both cases) for agent $1$ is $(90, <1,Ledger, x>, x, x)$ (where all but the second element are algorithm outputs). 
Any strategy $s_1$ response to this observed history will be the same for both nature-inputs, and thus the observed histories of the full run satisfy 
\begin{equation}
\label{eq:observed_s1_eq}
O_1(R_{\ninput', s_1}) = O_1(R_{\ninput'', s_1}).
\end{equation}
Equations~\ref{eq:observed_truth_neq},\ref{eq:observed_s1_eq} together contradict condition $(ii)$. 
\end{proof}

\begin{example}
\emph{max is $1$-NCC-vulnerable}

Consider agent $1$ with a strategy $s_1$ that upon a factual update for agent $1$, and given that there is a previous algorithm output and the last algorithm output is $\rho_v$,
updates with $<1,Ledger, \rho_v>$, i.e., the attacker repeats the last algorithm output as her own ledger update. Condition $(i)$ is satisfied: For the nature-input $\ninput = (<2, 100>, <1, 110>)$, the last algorithm output for the run with $s_1$ is $100$, while for the run with $truth_1$ it is $110$. Condition $(ii)$ is also satisfied: 
Consider two nature-inputs $\ninput, \ninput'$ that have the same observed run under $s_1$. Notice that the last algorithm output is the maximum over the other agents' truthful ledger updates. The last algorithm output under $truth_1$ is the maximum between other agents' ledger updates and agent $1$ maximum factual update, which is also observed under $s_1$. Therefore, the last algorithm output under $truth_1$ is determined by the observed history under $s_1$, and the natural way for the attacker to infer it is by taking the max over observed algorithm outputs and its own factual updates. 
\end{example}

We call such methods to construct the algorithm outputs under $truth_j$ out of the observed history $O_j$ an inference function.

\begin{definition}
\label{def:inference_function}
An \emph{inference function} is a function from observed histories $O_j$ to algorithm $\rho$ outputs. 
\end{definition}

\begin{lemma}
\label{lem:inference_function}
If there is an inference function $i_j$ so that for every run $R$ of nature-input $\ninput$ with $s_j$, $i_j(O_j(R)) = \rho$, where $\rho$ is the last algorithm output of the run of $\ninput$ with $truth_j$, then condition $(ii)$ holds for $s_j$. 
\end{lemma}
\begin{proof}
Assume by contradiction there are two nature-inputs $\ninput, \ninput'$ with the same $O_j$ when running with $s_j$. The nature-inputs must be of the same length $r$, otherwise, there would be a different amount of total factual updates, and thus either a different amount of algorithm updates not initiated by $j$ ledger updates, or a different amount of $j$ factual updates, both of which are observable. Let $\ninput_{\ell}, \ninput'_{\ell}$ be the nature-inputs of length $\ell$ that start the same as $\ninput, \ninput'$ but end after $\ell$ rounds. Since they have the same observed runs $O^{\ell}_j$ (parameterized by $\ell$), running with $true_j$ they must have $i_j(O^{\ell}_j)$ as the algorithm output after the round $\ell$. We conclude that all algorithm outputs identify for the two nature-inputs running with $truth_j$. The factual updates for $j$ also identify for both nature-inputs since $M^O_j$ identify, and since running with $truth_j$ the ledger updates by $j$ are a copy of the factual updates of $j$, they also identify for both nature-inputs. We conclude that the observable runs for both nature-inputs running with $truth_j$ identify, in compliance with condition $(ii)$. 
\end{proof}

\section{Technical Lemmas for the Strategy Templates}
\label{app:st_lemmas}

\sneakAttackImp*

\begin{proof}
We show that the condition to start attack (line $1$) was previously invoked by $s_j$ during the run of the continuous protocol iff $O_j$ contains three subsequent elements, $<j,Factual,U_{cond}>,<j,Ledger, U_{attack}>, \rho_{v1}$ for some $\rho_{v1}$: If the condition was invoked, then at that point the last element in $O_j$ was $<j,Factual, U_{cond}>$, and the agent updates with $<j,Ledger,U_{attack}>$, and finally the ledger updates all with some algorithm output $\rho_{v1}$. If it was not invoked before, then the condition to end attack (in line $3$) was not as well (as it depends on the condition to start attack being previously invoked). Therefore all ledger updates by $j$ are of the form of an algorithm output following some $<j,Ledger,U>$ after $<j,Factual,U>$. Since $U_{cond} \neq U_{attack}$, the pattern  $<j,Factual,U_{cond}>,<j,Ledger, U_{attack}>, \rho_{v1}$ can not appear.

We can thus use the above signature (together with the additional conditions given in line $1$) to decide whether to invoke the condition to start the attack. 

The condition to end attack is invoked iff $O_j$ last four elements are either of the form $$<j,Factual,U_{cond}>, <j,Ledger, U_{attack}>, \rho_{v1}, \rho_{v2},$$ for some algorithm outputs $\rho_{v1}, \rho_{v2}$, or of the form 
$$<j,Factual,U_{cond}>, <j,Ledger, U_{attack}>, \rho_{v1}, <j,Factual,U>,$$
for some algorithm output $\rho_{v1}$ and factual update $U$. We verify this signature matches the verbal description. If this signature appears, by our conclusion, the subsequent elements $<j,Factual,U_{cond}>, <j,Ledger, U_{attack}>, \rho_{v1}$ show that the condition to start attack was invoked. If we see a factual update or an algorithm output after that, it can only result in the continuous protocol from some agent receiving a factual update. Since these are the last elements in $O_j$, and there is no additional ledger update, the condition to end attack could not have previously been invoked since it only happens after the condition to start attack was invoked and sends an additional Ledger update. 

\end{proof}

\sneakAttackFormal*

\begin{proof}

For condition $(ii)$, consider two nature-inputs $\ninput, \ninput'$ with the same observed run 
\begin{equation}
\label{eq:ninput_cond}
    O_j(R_{\ninput,s_j}) = O_j(R_{\ninput', s_j}).
\end{equation}

If $O_j(R_{\ninput, s_j}) = O_j(R_{\ninput, truth_j})$, it means that the condition to start attack (line $1$) of $s_j$ was not invoked during the run. Thus, there is no update $<j, Factual, U_{cond}>$ in $O_j(R_{\ninput,s_j})$, and by Eq.~\ref{eq:ninput_cond} also not in $O_j(R_{\ninput',s_j})$. We therefore conclude that the condition to start attack is not invoked in the run of $\ninput'$ with $s_j$. Since the condition to end attack (line $3$) is only invoked if at a previous stage the condition to start attack was invoked, and so we conclude that all updates by $s_j$ for $\ninput'$ are truthful, and therefore $O_j(R_{\ninput',s_j}) = O_j(R_{\ninput', truth_j})$. All in all, the equations establish that $O_j(R_{\ninput, truth_j}) = O_j(R_{\ninput', truth_j})$. 

If $O_j(R_{\ninput, s_j}) \neq O_j(R_{\ninput, truth_j})$, then the condition to start attack must have been invoked during the run and there is a first update $U = <j, Factual, U_{cond}>$ in $O_j(R_{\ninput, s_j})$ such that the preceding algorithm output is $\rho_{cond}$. Let the index of this element be $i$. Before this update $U$, $s_j$ only responds truthfully, and so the observed runs satisfy $O_j(R_{\ninput, truth_j})_{1:i-1} = O_j(R_{\ninput, truth_j})_{1:i-1}$. Factual updates are preserved across observed runs with different strategies ($s_j$ vs $truth_j$), and so $O_j(R_{\ninput, truth_j})_i = O_j(R_{\ninput, s_j})_i = O_j(R_{\ninput',s_j})_i = O_j(R_{\ninput',truth_j})_i$. Since $truth_j$ follows each factual update of $j$ with a ledger update with the same $U$, we have $O_j(R_{\ninput, truth_j})_{i+1} = <j, Ledger, U_{cond}> = O_j(R_{\ninput', truth_j})_{i+1}$. Since we require that agent $j$ can infer the algorithm output $\rho_{infer}$ under $truth_j$ immediately after the start of the attack, and the observed histories up until this algorithm output identify for $R_{\ninput, s_j}, R_{\ninput',s_j}$, it must hold that $O_j(R_{\ninput, truth_j})_{i+2} = \rho_{infer} = O_j(R_{\ninput', truth_j})_{i+2}$.

We assume for simplicity that the first factual update after the factual update in index $i$ is an update of $j$. The argument can be extended to the case where it is not with more details. 

As noted before factual updates are preserved in the observed histories of different strategies, and so if there are at most $i+2$ elements in $O_j(R_{\ninput, truth_j})$, then the factual update that is element $i$ corresponds to the last element in $\ninput$. Therefore it is also the last factual update in $O_j(R_{\ninput, s_j})$, and by Eq~\ref{eq:ninput_cond} also in $O_j(R_{\ninput', s_j})$, and by the same argument in $O_j(R_{\ninput', truth_j})$. Since in runs with $truth_j$ each factual update of $j$ is followed exactly by a ledger update of $j$ and an algorithm output, we conclude that there are no more elements after $i+2$ for $O_j(R_{\ninput', s_j})$, and so it identifies with $O_j(R_{\ninput, s_j})$ (as we've shown all elements are the same). 

If there is a factual update at index $i+3$ in $O_j(R_{\ninput, truth_j})$, since we assume it is for $j$, it identifies for the two nature-inputs' observed histories with $s_j$, and as factual updates do not depend on strategy, we also have $O_j(R_{\ninput, truth_j})_{i+3} = O_j(R_{\ninput', truth_j})_{i+3}$. The subsequent ledger update and algorithm output thus identify as well. Note that by the condition to end attack, for both $\ninput, \ninput'$, the ledger update at step $i+4$ by $s_j$ is $U_{re-sync}$.  

At any step after $i+4$, the union of points sent throughout the ledger history identifies with the union of points in the factual history, by our requirement that $U_{re-sync} = U_{cond} \setminus U_{attack}$. All in all this shows that observed histories under $s_j$ identify $\implies$ observed histories under $truth_j$ identify, which is logically equivalent to condition $(ii)$.

\end{proof}

\triangulationAttackImp*

\begin{proof}
We specify the way to implement the required predicates, without giving the full proof, which goes by an argument similar to the proof structure of Lemma~\ref{sneakAttackImp}. 

For line $1$, there is a factual update after the last ledger update by agent $j$, iff it is either a factual update of $j$ or of another agent $a \neq j$. 

There is such factual update of an agent $a \neq j$ iff the last two elements in $O_j$ are both algorithm outputs, or there is only one element in $O_j$ and it is an algorithm output (this is the case where there are no ledger updates by agent $j$). The case where it is a factual update of agent $j$ is immediately visible in $O_j$ (agent $j$ can see its own factual updates). 

For line $2$, a triangulation attack is ongoing iff the pattern above is matched, followed by a series of pairs of the form $<j,Ledger, U>, \rho_{v1}$ for some ledger update $U$ and algorithm output $\rho_{v1}$. If the number of pairs is $i$, then the triangulation attack previously executed $i$ steps and we are at step $i+1$ of the attack. 

For line $3$, we reach it only given that $i$ is defined, and by the two possible signatures that make it happen (either starting a new triangulation attack or continuing an ongoing triangulation attack) assume that there are at least two algorithm outputs in $O_j$, hence it is valid to define $\rho_{i-1}$ as the last algorithm output in $O_j$. 

\end{proof}

\section{Missing Proofs for \texorpdfstring{$k$}---Center and \texorpdfstring{$k$}---Median}
\label{app:kcenter}

\kmedian*

\begin{proof}
$k$-median fits the definition of a set-choice algorithm. We show that it has forceable winners. We show it for $3$-median with the domain $R$, but the proof for general $k,R^d$ is similar. Let $x\in R$ and a set $S$ with $x\in S$. Let the symmetric completion of $S$ around $x$ be $O = \cup_{s\in S} \{s,2x-s\}$. I.e., every point $s = x + \epsilon$ is added the matching $s' = x - \epsilon$. Let $C = \max\{ \sum_{o\in O} |o - x|, 1\}$. Let $\bar{S} = O \cup \{x + 10C, x+100C\}$. For this construction, the following claim holds and completes the proof: 

\begin{claim}
$\rho(S \cup \bar{S}) = \{x, x+ 10C, x+100C\}$.
\end{claim}
\begin{proof}
If $\{x+ 10C, x+100C\} \subseteq \rho(S \cup \bar{S})$, then the remaining center $\zeta$ will have all remaining points closest to it. If we write $\zeta = x + \delta$, then the total cost attributed to this center must satisfy $|\delta| + \sum_{s\in S | \epsilon = s - x > 0} |(x + \epsilon) - (x + \delta)| + |(x - \epsilon) - (x + \delta)| = |\delta| + \sum_{s\in S | \epsilon = s - x > 0}(|-\epsilon - \delta| + |\epsilon - \delta|) \geq |\delta| + \sum_{s\in S | \epsilon = s - x > 0} 2\epsilon$. But the cost if $x$ is the remaining center is exactly $\sum_{s\in S | \epsilon = s - x > 0} 2\epsilon$, and so it is strictly better than the cost with $|\delta| > 0$. We conclude that in this case $x$ is the remaining center.

If either of $\{x+ 10C, x+100C\}$ is not in $\rho(S \cup \bar{S})$, then the cost of the solution is at least $9C$. But the cost for $\{x,x+ 10C, x+100C\}$ is exactly $C$, and $C > 0$.  
\end{proof}

\end{proof}

\kcenterPeriodic* 

The proof of the theorem is two-fold using the distinction of explicitly lying and omission strategies. As for explicitly-lying strategies, Claim~\ref{clm:condition2_forceable_domain} holds for periodic communication as well, with minor adjustments. We are thus left to show:

\begin{lemma}
An omission strategy $s_j$ for $k$-center violates condition $(ii)$ with the periodic protocol. 
\end{lemma}

\begin{proof}
In the proof of Corollary~\ref{cor:k_center} we show that $k$-center has forceable winners by a certain construction. We now use similar ideas to get a construction with more detailed properties, as formalized in the following claim:

\begin{claim}
\label{clm:k_center}
For $k$-center, for every set $S$ with $|S| \geq 2$ and a point $x \in S$, there is such $y \in S$, and $\bar{S}, \bar{S}'$ so that

$$\rho((S \cup \bar{S}') \setminus \{x\}) = \rho(S \cup \bar{S}') = \rho((S \cup \bar{S}) \setminus \{x\}) = \{y, \eta_1, \ldots, \eta_{k-1}\}$$

for some $\eta_1, \ldots, \eta_{k-1} \not \in S$ (in particular, not $x$). 

In addition, $\bar{S}$ satisfies the conditions of Definition~\ref{def:forceable_winners}.

\end{claim}
\begin{proof}
We show an explicit construction. 
Let $y = \arg\min_{x'\in S, x'\neq x} |x'-x|$. Let $\Delta = \max_{s\in S} |x-s|$. Let $\bar{S} = \{x + 2\Delta, x - 2\Delta, x+10\Delta, ..., x + 10^{k-1}\Delta, \bar{S}' = \{y + \Delta, y - \Delta, x+10\Delta, \ldots, x+10^{k-1}\Delta\}$. We have $\rho(S \cup \bar{S}) = \{x, x+10\Delta, \ldots, x+10^{k-1}\Delta\}, \rho((S \cup \bar{S}') \setminus \{x\}) = \rho(S \cup \bar{S}') = \rho((S \cup \bar{S}) \setminus \{x\}) = \{y, x + 10\Delta, \ldots, x+10^{k-1}\Delta\}$. 
\end{proof}

We now prove the lemma statement. Consider some nature-input $\ninput$ be the shortest (in terms of number of elements) where $s_j$ sends a ledger update with an explicit lie $x$, and let $L_R = L_j(R_{\ninput, s_j}), F_R = F_j(R_{\ninput, s_j})$ be the union of all ledger, factual updates respectively by $j$. Let $S = F_R \cup L_R \cup \{x + 1\}$, where we add the point $x+1$ to make sure there is some additional point besides $x$ in $S$ and have $|S| \geq 2$. 
 
Let $E_1 = (S \cup \bar{S}) \setminus \{x\}, E_2 = (S \cup \bar{S}_{\epsilon}) \setminus \{x\}$.
As $\rho(E_1 \cup \{x\}) \neq \rho(E_2 \cup \{x\})$, it must hold that $E_1 \neq E_2$. 
Let $r$ be the last round of $\ninput$. 
 Let $\ninput_1, \ninput_2$ be $\ninput$ with an additional last element $<i,E_1,r>, <i,E_2,r>$ respectively, for some agent $i\neq j$ (or, if all agents already have an element in this round, add $E_1, E_2$ respectively to the element of one of them). We have that $O_j(R_{\ninput_1,s_j}) = O_j(R_{\ninput_2,s_j})$, since the nature-inputs identify up until the last round, and the last round induces only an algorithm output observation by $j$, which satisfies $\rho(L_R \cup E_2) = \rho((S \cup \bar{S}_{\epsilon}) \setminus \{x\})) \stackrel{Claim~\ref{clm:k_center}}{=}  \rho((S \cup \bar{S}) \setminus \{x\})) = \rho(L_R \cup E_1).$

However, the last algorithm output in $O_j(R_{\ninput_1,truth_j})$ is $\rho(F_R \cup E_1) = \rho(F_R \cup (F_R \cup L_R \cup \bar{S})) = \rho(S \cup \bar{S})$, and thus has the element $x$ (By Claim~\ref{clm:k_center} guarantee that $\bar{S}$ satisfies the conditions of Definition~\ref{def:forceable_winners}). On the other hand, the last algorithm output in $O_j(R_{\ninput_2,truth_j})$ is $\rho(F_R \cup E_2) = \rho(S \cup \bar{S}_{\epsilon})$, and does not contain $x$ as an element by Claim~\ref{clm:k_center}.
\end{proof}

\section{Missing Proofs for \texorpdfstring{$d$}---Linear Regression}
\label{app:lr}

\begin{example}
\label{ex:LR_sneak_attack}

\emph{$1-LR$ is $1$-NCC-vulnerable using an explicitly-lying sneak attack}: Use Algorithm~\ref{alg:sneak_attack_template} with \[
\begin{split}& U_{cond} = (\mathbf{X^{cond}}, \mathbf{y^{cond}}) = (\begin{bmatrix} 1 & 3 \\
1 & 0 \\
1 & 0 \end{bmatrix},\begin{bmatrix} 1 \\
1 \\
1 \end{bmatrix}) , U_{attack} = (\mathbf{X^{attack}}, \mathbf{y^{attack}}) = (\begin{bmatrix} 1 & 2 \end{bmatrix},\begin{bmatrix} 2 \end{bmatrix}), \\
& U_{re-sync} = (\mathbf{X^{re-sync}}, \mathbf{y^{re-sync}}) = (\begin{bmatrix} 1 & 2 \\
1 & -1 \end{bmatrix},\begin{bmatrix} 0 \\
1\end{bmatrix}), \rho_{cond} = \beta = \begin{bmatrix} 1 \\
0 \end{bmatrix}.\end{split}\] 

Condition $(i)$ is satisfied since for nature-input $\ninput = (<1, (\begin{bmatrix} 1 & 1 \\
1 & 0 \end{bmatrix},\begin{bmatrix} 1 \\
1\end{bmatrix})>, <2, U_{cond}>)$, the run with $truth_2$ yields algorithm outputs $\begin{bmatrix} 1 \\
0 \end{bmatrix}, \begin{bmatrix} 1 \\
0 \end{bmatrix}$ but the run with $s_2$ yields $\begin{bmatrix} 1 \\
0 \end{bmatrix}, \begin{bmatrix} \frac{5}{6} \\
\frac{1}{2} \end{bmatrix}$. 

For condition $(ii)$, the argument generally follows the proof of Lemma~\ref{lem:sneak_attack_formal}. We note two important distinctions:

\begin{itemize}
    \item Given $\rho_{cond}, U_{cond}$ we can infer the algorithm output under $truth_j$ is $\rho_{cond}$. That is since $\rho_{cond}$ has the minimal cost function given all updates previous to $U_{cond}$ (for both $\ninput, \ninput'$): We know that since it is the algorithm output before the factual update $U_{cond}$. Moreover, it has cost $0$ with regards to $U_{cond}$.
    
    \item At any step after $i+4$ (the completion of the sneak attack), the union of points sent throughout the ledger history does not identify anymore with the union of points in the factual history, since  the ledger history contains explicit lies (namely, any of the points in our choice of $U_{attack}, U_{re-sync}$). However, for the calculation of the algorithm output in $1-LR$, we have $\rho = (X^T X)^{-1} X^T y$. Since updates aggregation is additive, as long as two different updates have the same $X^T X, X^T y$, any sequence of updates containing them would have the same algorithm outputs. In our case, we have \[
    \begin{split}
        & (X^{cond})^T X^{cond} = \begin{bmatrix} 3 & 3 \\ 3 & 9 \end{bmatrix} = (X^{attack})^T X^{attack} + (X^{re-sync})^T X^{re-sync} , \\
        & (X^{cond})^T y = \begin{bmatrix} 3 \\ 3\end{bmatrix} = (X^{attack})^T y^{attack} + (X^{re-sync})^T y^{re-sync}.
        \end{split}
        \] By the construction of $s_j$, any sequence of updates after the attack is ``re-synced'' behaves just as if the agent has acted truthfully, and so the observed truthful histories identify subsequently. All in all this shows that observed histories under $s_j$ identify $\implies$ observed histories under $truth_j$ identify, which is logically equivalent to condition $(ii)$. 
    
\end{itemize}

\end{example}

\dlrPointLine*
\begin{proof}
Let $\rho = \begin{bmatrix} \alpha_1 \\ \ldots \\ \alpha_{d+1} \end{bmatrix}$, and let the algorithm output after adding the point $X^*, y^*$ be  $\rho' = \begin{bmatrix} \alpha_1' \\ \ldots \\ \alpha_{d+1}' \end{bmatrix}$. First, we show that it can not be that $\rho = \rho'$. Assume otherwise, then by the extremal condition over the optimization function, 
 $$\alpha_1 = \frac{\sum_{i=1}^n (y_i - \sum_{j=1}^d \alpha_{j+1} x^j_i)}{n},$$ and similarly for $\rho' = \rho$, 
\[
\begin{split}
& \alpha_1 = \frac{(\sum_{i=1}^n y_i - \sum_{j=1}^d \alpha_{j+1} x^j_i) + (y^* - \sum_{j=1}^d \alpha_{j+1} x^*_j)}{n+1} = \\
& \frac{n\alpha_1 + \alpha_1}{n+1} + \frac{y^* - \sum_{j=1}^d \alpha_{j+1} x^*_j - \alpha_1}{n+1} = \\
& \alpha_1 + \frac{y^* - \sum_{j=1}^d \alpha_{j+1} x^*_j - \alpha_1}{n+1} \neq \alpha_1,
\end{split}
\]
where the last inequality is since the point $(X^*,y^*)$ is outside the line and thus $y^* - \sum_{j=1}^d \alpha_{j+1} x^*_j - \alpha_1 \neq 0$. We arrived at a contradiction and so we may subsequently assume $\rho \neq \rho'$. Now, again assume by contradiction that the lines intersect at the point $(X^*, y^*)$. Then, with respect to $(X^*, y^*)$ the two lines would have the same cost function value 0, but overall with respect to all $y_i$ with $1 \leq i \leq n$, $\rho$ is the unique optimal cost minimizer, so we conclude that $\rho$ has lower cost overall than $\rho'$ with respect to all the given points, in contradiction to $\rho'$ being optimal. 
\end{proof}

\end{document}